\RequirePackage{fix-cm}
\documentclass[11pt]{article}

\usepackage[T1]{fontenc}
\usepackage{lmodern}
\rmfamily 
\DeclareFontShape{T1}{lmr}{b}{sc}{<->ssub*cmr/bx/sc}{}
\DeclareFontShape{T1}{lmr}{bx}{sc}{<->ssub*cmr/bx/sc}{}
\usepackage{silence}
\usepackage{amsmath,amssymb,amsthm,amsfonts,latexsym,bbm,xspace,thm-restate}
\usepackage{graphicx,float,mathtools,braket}
\usepackage{enumitem,booktabs,forest,mathdots,soul,comment}
\usepackage[useregional]{datetime2}
\DTMusemodule{english}{en-US}

\usepackage[letterpaper,margin=1in]{geometry}
\usepackage[colorlinks,citecolor=blue,,linkcolor=magenta,bookmarks=true]{hyperref}
\usepackage[nameinlink,capitalize]{cleveref}

\newtheorem{theorem}{Theorem}[section]
\newtheorem{lemma}[theorem]{Lemma}

\newtheorem{fact}[theorem]{Fact}
\crefname{fact}{Fact}{Facts}

\theoremstyle{definition}
\newtheorem{definition}[theorem]{Definition}

\newtheorem{remark}[theorem]{Remark}

\DeclareMathOperator{\tr}{tr}

\DeclareMathOperator{\poly}{poly}

\DeclareMathOperator{\sgn}{sgn}

\DeclareMathOperator{\diag}{diag}

\newcommand{\lambdamin}{\lambda_{\min}}

\newcommand{\ketbra}[2]{\ket{#1}\!\bra{#2}}

\DeclareMathOperator{\SWAP}{SWAP}

\newcommand{\eps}{\varepsilon}

\newcommand{\YES}{\text{YES}}
\newcommand{\NO}{\text{NO}}

\newcommand{\nth}{^\text{th}}

\newcommand{\reals}{\mathbb R}

\newcommand{\class}[1]{\ensuremath{\mathsf{#1}}\xspace}
\mathchardef\mhyphen="2D 

\newcommand{\PTIME}{\class{P}}
\newcommand{\NP}{\class{NP}}

\newcommand{\QMA}{\class{QMA}}

\newcommand{\StoqMA}{\class{StoqMA}}

\newcommand{\prob}[1]{\textup{\textsc{#1}}\xspace}
\newcommand{\kLH}[1][k]{\ensuremath{#1}\prob{-LH}}

\newcommand{\sLH}[1][\calS]{\ensuremath{#1}\prob{-LH}}
\newcommand{\prodLH}{\prob{prodLH}}
\newcommand{\sprodLH}[1][\calS]{\ensuremath{#1\mhyphen}\prob{prodLH}}

\newcommand{\threeColoring}{\prob{3-Coloring}}

\newcommand{\qmc}{\prob{QMC}}
\newcommand{\qmcLong}{\prob{Quantum Max-Cut}}

\newcommand{\qmcProd}{\prob{prodQMC}}

\newcommand{\maxcut}{\prob{Max-Cut}}
\newcommand{\maxcutop}{\operatorname{MC}}

\newcommand{\kmaxcut}[1][k]{\ensuremath{\prob{MC}_{#1}}}
\newcommand{\kmaxcutLong}[1][k]{\ensuremath{\prob{Max-Cut}_{#1}}}
\newcommand{\kmaxcutop}[1][k]{\operatorname{MC}_{#1}}

\newcommand{\Wlinearmaxcut}[1][W]{\ensuremath{\prob{MC}^{\scriptscriptstyle\prob{L}}_{#1} }}
\newcommand{\WlinearmaxcutLong}[1][W]{\ensuremath{#1\!\prob{-linear-Max-Cut}}}
\newcommand{\Wlinearmaxcutop}[1][W]{\ensuremath{\operatorname{MC^{\scriptscriptstyle L}_{#1}} }}

\newcommand{\calC}{\mathcal{C}}

\newcommand{\calF}{\mathcal{F}}

\newcommand{\calS}{\mathcal{S}}

\DeclarePairedDelimiter\absd{\lvert}{\rvert}
\DeclarePairedDelimiter\normd{\lVert}{\rVert}
\DeclarePairedDelimiter\parend{\lparen}{\rparen}

\DeclarePairedDelimiter\cbracd{\lbrace}{\rbrace}
\DeclarePairedDelimiter\brac{\lbrack}{\rbrack}

\newcommand{\abs}[1]{\absd*{#1}}
\newcommand{\norm}[1]{\normd*{#1}}
\newcommand{\paren}[1]{\parend*{#1}}

\newcommand{\cbrac}[1]{\cbracd*{#1}}

\newcommand{\bO}[1]{\operatorname*{O}\paren{#1}}

\newcommand{\bOm}[1]{\operatorname*{\Omega}\paren{#1}}

\newcommand{\bT}[1]{\operatorname*{\Theta}\paren{#1}}

\newcommand{\eg}{e.g.\xspace}

\title{Complexity Classification of Product State Problems for Local Hamiltonians}

\author{
    John Kallaugher\\\small{\textsl{Sandia National Laboratories}}\\\small{\texttt{\href{mailto:jmkall@sandia.gov}{jmkall@sandia.gov}}}\and
    Ojas Parekh\\\small{\textsl{Sandia National Laboratories}}\\\small{\texttt{\href{mailto:odparek@sandia.gov}{odparek@sandia.gov}}}\and
    Kevin Thompson\\\small{\textsl{Sandia National Laboratories}}\\\small{\texttt{\href{mailto:kevthom@sandia.gov}{kevthom@sandia.gov}}}\and
    Yipu Wang\\\small{\textsl{Sandia National Laboratories}}\\\small{\texttt{\href{mailto:yipu.wang@gmail.com}{yipu.wang@gmail.com}}}\and
    Justin Yirka\\\small{\textsl{Sandia National Laboratories}}\\\small{\textsl{The University of Texas at Austin}}\\\small{\texttt{\href{mailto:yirka@utexas.edu}{yirka@utexas.edu}}}
}

\begin{document}
\date{}

\maketitle

\begin{abstract}
\noindent
Product states, unentangled tensor products of single qubits, are a ubiquitous
ansatz in quantum computation, including for state-of-the-art Hamiltonian
approximation algorithms.  A natural question is whether we should expect to
efficiently solve product state problems on any interesting families of
Hamiltonians.

We completely classify the complexity of finding minimum-energy product states
for Hamiltonians defined by any fixed set of allowed 2-qubit interactions.
Our results follow a line of work classifying the complexity of solving
Hamiltonian problems and classical constraint satisfaction problems
based on the allowed constraints. We prove
that estimating the minimum energy of a product state is in $\PTIME$ if and
only if all allowed interactions are 1-local, and $\NP$-complete otherwise.
Equivalently, any family of non-trivial two-body interactions generates
Hamiltonians with $\NP$-complete product-state problems.  Our hardness
constructions only require coupling strengths of constant magnitude.

A crucial component of our proofs is a collection of hardness results for a new
variant of the \prob{Vector Max-Cut} problem, which should be of independent
interest.  Our definition involves sums of distances rather than squared
distances and allows linear stretches.

We similarly give a proof that the original \prob{Vector Max-Cut} problem is
$\NP$-complete in 3 dimensions.  This implies that optimizing over product
states for \prob{Quantum Max-Cut} (the quantum Heisenberg model) is
$\NP$-complete, even when every term is guaranteed to have positive unit
weight.
\end{abstract}

\section{Introduction}\label{sec:intro}
Product states, unentangled tensor products of single-qubit states, have served
as an effective focus for better understanding quantum phenomena.
Because general quantum states cannot be described efficiently, approximation algorithms
must be restricted to output some subset of states, an ansatz.
Mean-field
approaches are common as first steps in statistical mechanics, and recent
approximation algorithms for extremal energy
states of local Hamiltonians have relied on proving that product
states provide good approximations in particular regimes \cite{gharibian2012approximation,harrowbrandao2016product,bravyi2019approximation,
	gharibian_et_al:LIPIcs:2019:11246,parekh_et_al:LIPIcs.ESA.2021.74,parekh2022optimal}.
In fact, because in some natural regimes the ground states are rigorously
well-approximated by product states \cite{harrowbrandao2016product},
optimal approximation algorithms for local Hamiltonians on arbitrary interaction graphs
must be capable of finding good product states.
Understanding product state optimization is essential for understanding the
complexity of Hamiltonian approximation generally.

Product states are a natural intermediate between classical and quantum states,
allowing for superposition but not entanglement.  Unlike general quantum
states, they have succinct classical descriptions: a single-qubit pure state
can be specified by two complex numbers, and an $n$-qubit pure product state
can be specified by $2n$ complex numbers.  One could consider ``more quantum''
intermediates, in the form of reduced states of two or more
qubits.  However, verifying the consistency of a set of quantum marginals is a
$\QMA$-complete problem, even for 2-qubit reduced states
\cite{liu2006consistency,broadbent2022qma}.  Therefore, product states are
uniquely useful when optimizing directly over state vectors.

We study the following question: for a family of Hamiltonians defined by
a given set of allowed interactions, what is the complexity of computing the
extremal energy over product states?
Additionally, how does the complexity of optimizing over product states
relate to that of optimizing over general states?
For example, for a $\QMA$-hard local Hamiltonian, must finding the optimal product state in turn be
$\NP$-hard?\footnote{The product state problem is always in $\NP$ since product
states have succinct classical descriptions with which we can compute the expected
energy contribution from each local Hamiltonian term in time polynomial in the size of the term.}

This question follows a long line of work classifying the
complexity of constraint satisfaction problems (CSPs) based on the sets of
allowed constraints, clauses, or interactions between variables.
In particular, the dichotomy theorem of Schaefer \cite{schaefer1978complexity}
showed that for any set of allowed Boolean constraints, the family of CSPs is either
efficiently decidable or is $\NP$-complete.
In the context of quantum problems, Cubitt and Montanaro \cite{CM16_hamiltonians}
introduced a similar classification of ground state energy problems for 2-local Hamiltonians,
showing that for any fixed set of allowed 2-qubit interactions,
the \kLH[2] problem is either in $\PTIME$ or $\NP$-, $\StoqMA$-, or $\QMA$-complete
(the $\StoqMA$ case relies on the concurrent work of
Bravyi and Hastings \cite{bravyi2017complexity}).
We briefly survey some of this line of classical and quantum work
in \nameref{ssec:relatedWork} below.

While the complexity of finding the extremal states (e.g.\ the ground state) of
2-local Hamiltonians is well understood, the complexity of finding optimal
product state solutions has been only sparsely studied~\cite{hwang2023unique}.
The only $\NP$-hardness results for such problems are based on mapping a
classical problem to a diagonal 0-1 valued Hamiltonian \cite{wocjan20032}.

An additional motivation for our study is the hope of developing new methods for
identifying families of local Hamiltonians
for which problems involving \emph{general} ground states are \emph{not} hard.
While a complete complexity classification for the general $\kLH[2]$ problem is
known, more refined attempts at classification
which take into account restrictions on the sign of the weights or geometry
of the system are currently incomplete~\cite{PM17_hamiltonians}.
Developing algorithms for product states is a ``mostly classical'' problem that
is easier to analyze, and progress involving product states may inform our
expectations regarding general states.

In this work, we completely classify the complexity of finding optimal product
states for families of 2-local Hamiltonians.  In fact, we find the complexity
of the product state problem is fully determined by the complexity of the
general local Hamiltonian problem: if the general problem is $\NP$-hard, the
product state problem is $\NP$-complete, and otherwise it is in $\PTIME$.  To
arrive at our results, we study a variant of the \prob{Vector Max-Cut} which
should be of independent interest especially to the optimization community.  As
a corollary to our classification theorem, we give the first published proof
that estimating optimal product state energies in the \qmcLong{} model is
$\NP$-complete,
and we show hardness holds even for unweighted Hamiltonians.\footnote{
	A more complex unpublished proof based on large graph cycles was known
	earlier by Wright~\cite{johnWright}.}

\subsection{Our Contributions}
Formal definitions are given in \cref{sec:prelim}.  A $k$-local Hamiltonian is
a sum of Hamiltonian terms each of which only acts non-trivially on at most $k$
qubits, analogous to $k$-variable Boolean clauses.  $\kLH$ denotes the problem
of estimating the ground state energy (the minimum eigenvalue across all
states) of a $k$-local Hamiltonian to inverse-polynomial additive precision.
Given a set of local terms $\calS$, \sLH{} is $\kLH$ restricted to Hamiltonians
such that every term belongs to $\calS$.  Finally, $\prodLH$ and $\sprodLH{}$
are the restrictions of these problems to product states, i.e.\ to minimize
$\braket{\phi|H|\phi}$ where $H$ is the Hamiltonian and $\ket{\phi}$ ranges
over tensor products of single-qubit states.\footnote{
	An earlier version of this work considered \emph{exact} versions of
	product state and graph problems.
	We have now improved our hardness results to hold up to an
	inverse-polynomial additive gap.
	}

\paragraph{2-local $\sprodLH$}
The classification of the general ground state energy problem by Cubitt and
Montanaro \cite{CM16_hamiltonians} completely classifies \sLH{} for any fixed
set $\calS$ of 2-qubit terms, showing it is either in $\PTIME$ or it is one of
$\NP$-, $\StoqMA$-, or $\QMA$-complete.

 In the same vein, we give a complete classification of product state complexity for families
 of 2-local Hamiltonians as a function of the set of
 allowed 2-qubit interactions.
 For any given set $\calS$ of 2-qubit terms, we prove the problem
 \sprodLH{} is either
 in $\PTIME$ or is $\NP$-complete.
 To the best of our knowledge, ours is the first systematic
 inquiry into the complexity of product state problems.

\begin{restatable}{theorem}{prodStatesHard}
\label{thm:prodStatesHard}
    For any fixed set of 2-qubit Hamiltonian terms $\calS$, if every matrix in
    $\calS$ is 1-local then $\sprodLH$ is in $\PTIME$, and otherwise $\sprodLH$
    is $\NP$-complete.
\end{restatable}

Additionally, our hardness constructions only require coupling strengths (weights)
of at most constant magnitude.
This is preferable in practice and contrasts with
most known $\QMA$-hardness constructions.

The sets for which \cite{CM16_hamiltonians} shows \sLH{} is in $\PTIME$ are the same for which we show \sprodLH{}
is in $\PTIME$, i.e.\ those containing only 1-local terms.
This immediately implies that a family of 2-local Hamiltonians has
efficiently-computable minimum product state energy if and only if it has
efficiently-computable ground state energy.

\begin{restatable}{corollary}{fLHeasyProdLHeasy}
  \label{corr:fLHeasyProdLHeasy}
  For any fixed set of 2-qubit Hamiltonian terms $\calS$, the problem $\sLH$ is
  in $\PTIME$ if and only if $\sprodLH$ is in $\PTIME$.
\end{restatable}
\vspace{-\topsep}
\begin{restatable}{corollary}{fLHhardProdLHhard}
  \label{corr:fLHeasyProdLHhard}
  For any fixed set of 2-qubit Hamiltonian terms $\calS$, the problem $\sLH$ is
  $\NP$-hard if and only if $\sprodLH$ is $\NP$-complete.
\end{restatable}

Our results imply that hardness of product state approximations is not restricted to
Hamiltonians for which product states well approximate ground states: for
any $\QMA$-hard family of terms, our result implies that (assuming $\QMA \not=
\NP$) we can construct a family of local Hamiltonians that are $\NP$-hard to
product state approximate \emph{and} for which the product states do not well
approximate the ground state by constructing Hamiltonians on two systems, each with one
of these properties, and taking a disjoint union thereof. This implies that
algorithms using product states to approximate the ground states of $\QMA$-hard
Hamiltonians face a ``double penalty'': hardness of approximating product
states which themselves imperfectly approximate ground states.

\paragraph{The Stretched Linear Vector Max-Cut problem}
Our hardness constructions for product state problems
embed an objective function
which we prove is $\NP$-complete.
This objective function generalizes the classical \maxcut{} problem.
Given work on other variations of \maxcut{}, we expect this problem
and our reductions should be of
independent interest, especially to the
optimization and approximation communities.

In \maxcut{}, one is given a graph $G=(V,E)$ and asked to assign to each vertex $v$ a label $\hat{v} = \pm 1$ so as to achieve the maximum number of oppositely labeled adjacent vertices:
\begin{equation}\label{eqn:MC}
	\operatorname{MC}(G) = \frac{1}{2}\max_{\hat{\imath}=\pm 1} \sum_{ij\in E}
	\parend*{1-\hat{\imath}\hat{\jmath}} = \frac{1}{2}\max_{\hat{\imath}=\pm 1}
	\sum_{ij\in E} \abs{\hat{\imath} - \hat{\jmath}} .
\end{equation}
A problem referred to as \prob{Vector Max-Cut}, \prob{Rank-$k$-Max-Cut}, or \kmaxcutLong{} (\kmaxcut{}) has been studied
\cite{briet2010positive,briet2011generalized,hwang2023unique}
which generalizes \maxcut{} to assigning $k$-dimensional unit vectors so as to maximize the angles between adjacent vertex labels, or equivalently to maximize the squared distances between adjacent vertex labels:
\begin{equation}\label{eqn:MCk}
\kmaxcutop(G) = \frac{1}{2}\max_{\hat{\imath}\in S^{k-1}} \sum_{ij\in E}
\parend*{1-\hat{\imath}\cdot \hat{\jmath}} = \frac{1}{4} \max_{\hat{\imath}\in
S^{k-1}} \sum_{ij\in E} \norm{\hat{\imath}-\hat{\jmath}}^2 .
\end{equation}
Our new problem can be seen as a stretched and linear version of \kmaxcut{}.
In \WlinearmaxcutLong{} (\Wlinearmaxcut{}), the goal is to assign unit vectors
so as to maximize the distance between adjacent labels:
\begin{equation}\label{eqn:WMC_intro}
    \Wlinearmaxcutop(G) = \frac{1}{2} \max_{\hat{\imath} \in S^{d-1}} \sum_{ij\in E} \norm{W\hat{\imath}-W\hat{\jmath}} = \frac{1}{2} \max_{\hat{\imath} \in S^{d-1}} \sum_{ij\in E} \sqrt{\norm{W\hat{\imath}} + \norm{W\hat{\jmath}} - 2(W\hat{\imath})^{\top}(W\hat{\jmath})} ,
\end{equation}
where $W$ is a fixed $d\times d$ diagonal matrix.  Comparing $\kmaxcut$ and
$\Wlinearmaxcut$, our problem sums over un-squared distances and incorporates a
linear stretch given by $W$. We consider the decision version of this problem,
in which the objective is to test whether the optimal solution is at least $b$
or no more than $a$, for $b - a \ge 1/\poly(n)$. Note that, unlike $\sLH$, this
is an unweighted problem---one could naturally define a weighted version but
our hardness results will not require this.

Geometrically, \kmaxcut{} corresponds to embedding a graph into the surface of
a unit sphere with the objective of maximizing the sum of the squared lengths
of every edge.  Likewise, our problem \Wlinearmaxcut{} corresponds to embedding
a graph into the surface of a $d$-dimensional ellipsoid, with radii defined by
the entries of $W$, with the objective of maximizing the sum of the (non-squared)
edge lengths.

Despite being generalizations of the $\NP$-complete \maxcut{} problem, hardness
of neither $\kmaxcut$ nor $\Wlinearmaxcut$ is trivial.  The Goemans-Williamson
approximation algorithm for \maxcut{} on an $n$-vertex graph begins with
efficiently computing the solution to $\kmaxcut[n]$ via an SDP.  In fact,
deciding $\kmaxcut$ is known to be in $\PTIME$ for any $k = \Omega(\sqrt{\abs{V}})$,
\cite[Theorem 8.4]{lovasz2003} or \cite[(2.2)]{barvinok1995problems}.  And
while it has been conjectured by Lov\'{a}sz that $\kmaxcut$ is $\NP$-complete
for all constants $k\geq1$, in \cite[p. 236]{lovasz2019graphs} and
earlier, no proof has been given for any $k>1$.

Our main theorem concerning \WlinearmaxcutLong{}, which is used to prove
\cref{thm:prodStatesHard},
is the following.

\begin{restatable}{theorem}{WlinMaxCutTheorem}
\label{thm:WlinMaxCut}
	For any fixed non-negative $W=\diag(\alpha,\beta,\gamma)$ with at least one of $\alpha,\beta,\gamma$ nonzero,
	\WlinearmaxcutLong{} is $\NP$-complete.
\end{restatable}

\paragraph{Quantum Max-Cut Product States and \kmaxcut[3]}
As a corollary of our classification theorem, we give the first published proof
of the fact that product state optimization in the \qmcLong{} (\qmc{}) model is
$\NP$-hard.  This model, also known as the anti-ferromagnetic Heisenberg model,
is equivalent to \sLH with $\calS=\{XX+YY+ZZ\}$.  We note that a sketch of a
different proof for this specific problem was previously known but
unpublished~\cite{johnWright}.  That proof was based on large graph cycles, and
our gadgets are simpler to analyze.

However, the proof of \cref{thm:prodStatesHard} utilizes Hamiltonian gadgets
involving negative weights (unlike the aforementioned proof
of~\cite{johnWright}).  This leaves open whether \qmcProd{} remains $\NP$-hard
on unweighted graphs. In \cref{sec:qmc}, we give a direct proof of hardness
using the fact that the unweighted product-state version of $\qmc$ is
equivalent to $\kmaxcut[3]$ (\cref{eqn:MCk}). Our work then is also the first
published proof that $\kmaxcut[k]$ is $\NP$-complete for some $k > 1$ (in our
case, $k=3$), partially resolving a conjecture of Lov\'{a}sz \cite[p.
236]{lovasz2019graphs}. Note that as with \Wlinearmaxcut, we consider the
decision version in which the goal is to determine whether the value is above
$b$ or below $a$, for $b$ and $a$ with inverse-poly separation.

\begin{restatable}{theorem}{maxCutthreeNPcomplete}\label{thm:3MaxCutNPcomplete}
  \kmaxcut[3] is $\NP$-complete.
\end{restatable}

\begin{restatable}{corollary}{qmcProdHard}\label{thm:qmcProdHard}
  \qmcLong{} restricted to product states, $\qmcProd$, is $\NP$-complete, even
  when all terms are restricted to have positive unit weight.
\end{restatable}

\subsection{Proof Overview}

\paragraph{2-local $\sprodLH{}$}
As product states have classical descriptions and their energies can be
calculated in polynomial time, $\sprodLH$ is automatically in $\NP$, so we
focus on how we show hardness.  Our approach is in two parts.  We show how to
reduce \Wlinearmaxcut{} to \sprodLH, and later we show that \Wlinearmaxcut{}
$\NP$-complete (\cref{thm:WlinMaxCut}).  More precisely, the first part of our
approach is to show that for any $\calS$ containing a strictly 2-local term,
there exists a corresponding weight matrix $W$ meeting the conditions of
\cref{thm:WlinMaxCut} so that \Wlinearmaxcut{} is $\NP$-hard.  For any instance
of \Wlinearmaxcut{} with this fixed $W$, we show how to construct a Hamiltonian
from $\calS$ such that the minimum product state energy encodes the
\Wlinearmaxcut{} value, yielding \cref{thm:prodStatesHard}.

We interpret problems over product states
as optimization problems on the collection of Bloch vectors for
each single-qubit state.
For example, consider $\sprodLH{}$ with the specific set $\calS =
\lbrace XX + YY + ZZ \rbrace$ (the \qmc{} model). In this case, by writing each
qubit $v$ in the Bloch vector representation
\[
	\ketbra{\phi_v}{\phi_v} = \frac{1}{2} \left( I + v_1 X + v_2 Y + v_3 Z \right) ,
\]
the energy contributed by an interaction between qubits $u$ and $v$ is
\begin{equation}\label{eq:XYZenergy}
	\tr\left(\left(X^uX^v + Y^uY^v +
	Z^uZ^v\right)\ketbra{\phi_u}{\phi_u}\ketbra{\phi_v}{\phi_v} \right) =
	u_1v_1 + u_2v_2 + u_3v_3 .
\end{equation}
So, given a Hamiltonian which is the sum of XYZ interactions between pairs of qubits,
the problem of estimating the extremal product state energies
is equivalent to optimizing the objective function
\[
	\sum_{uv} w_{uv} u \cdot v
\]
over 3-dimensional unit vectors, where each edge $uv$ corresponds to
the Hamiltonian's (weighted) interaction graph.
Up to constant shifts and scaling, this is equivalent to \kmaxcut[3],
introduced in \cref{eqn:MCk}.

More generally, because the Pauli matrices are a basis for all Hermitian matrices,
any 2-qubit interaction can be written as
$H = \sum_{i,j=1}^3 M_{ij} \sigma_i \sigma_j + \sum_{k=1}^3 \paren{c_k\sigma_k I + w_k I \sigma_k}$,
where the $\sigma_i$ are the Pauli matrices, for some $M$, $c$, $w$. Then, the
energy of a given product state is calculated similarly to \cref{eq:XYZenergy}.
However, the resulting expression potentially contains many terms.

Our approach is to take an arbitrary 2-qubit term and
insert it into gadgets which simplify the energy calculations.
First, in the proof of \cref{thm:prodStatesHard}, we borrow a trick of \cite{CM16_hamiltonians}
and symmetrize the terms.
For any 2-qubit interaction $H^{ab}$ on qubits $a,b$, the combined interaction $H^{ab}+H^{ba}$
is symmetric, invariant under swapping the qubits.
A similar trick handles the case of anti-symmetric terms.

Second, in the proofs of
\cref{lem:antisymProdLH_posProdLH,lem:symmetricProdLH}, we show how to use a
symmetric or anti-symmetric term to embed the \Wlinearmaxcut{} value into the
minimum energy of a gadget.  We begin by removing the 1-local terms, such as
$\sigma_1 I$ or $I\sigma_2$, again taking inspiration from gadgets used by
\cite{CM16_hamiltonians}.  For two qubits $u,v$ corresponding to two vertices
in a \Wlinearmaxcut{} instance, the gadget adds two ancilla qubits and weights
each interaction within the gadget to effectively cancel out the 1-local terms.
When the two ancilla qubits vary freely, we find the minimum energy contributed
by the entire four-qubit gadget is determined by the distance between the
states of $u$ and $v$.  Although each individual edge contributes energy
proportional to the squared distance between their states, the overall gadget
contributes energy proportional to just the distance of the two ``vertex
qubits'', ${-\norm{Mu-Mv}}$.  With some massaging, we can treat $M$ as a
non-negative diagonal matrix which meets the conditions of
\cref{thm:WlinMaxCut}.

Therefore, as desired, we have that the $\Wlinearmaxcut$ value for an
$\NP$-complete instance of \Wlinearmaxcut{} can be embedded into the minimum product state
energy of an \sprodLH{} instance.

\paragraph{Stretched Linear Vector Max-Cut}
To prove that $\Wlinearmaxcut$ is $\NP$-hard for any fixed diagonal
non-negative $W$ with at least one nonzero entry, we divide into three cases,
in \cref{lem:WlMC_twoMax,lem:WlMC_allEqual,lem:WlMC_oneMax}.  Our first proof
is a reduction from the standard \maxcut{} problem, while the other two are
reductions from \threeColoring.

When there is a unique largest entry of $W$, we reduce from \maxcut{} by taking
any input graph $G$ and forming $G'$ by adding large star gadgets around each of
the vertices of $G$, each using many ancilla vertices.  Because the ancilla
vertices in each gadget have just one neighbor (the original vertex at the
center of the gadget), their optimal vector labels given any choice of labels
for the center vertices are the negation of the center vertex labels.  This
means they heavily penalize assigning the center vertices any labels that are
not along the highest-weight axis.  Therefore, when the maximum entry of $W$ is
unique, the optimal \Wlinearmaxcut{} assignment to $G'$ will have its vector
labels almost entirely along the highest-weight axis. The assignment can trivially earn the
maximum possible value on the star gadgets, and the amount additional amount it can earn on the original edges of $G$ corresponds to the \maxcut{} value of $G$.

When all of the entries of $W$ are equal, we reduce from \threeColoring{}.
Given a graph $G$, we construct $G'$ by replacing each edge with a $4$-clique
gadget, made by adding one ancilla vertex per gadget, along with a single ancilla vertex shared by every gadget.
We show that $G$ is 3-colorable iff $G'$ has a sufficiently large
$\Wlinearmaxcut$ value.  Specifically, we show this holds iff there is a vector
assignment that simultaneously achieves (nearly) the maximum value on all of
the clique gadgets.  Achieving the maximum objective value on a clique
corresponds to maximizing the total distance between each pair of vectors, and
this enforces a predictable arrangement.

When weights are equal, assigning these vectors can be viewed as inscribing
vectors in the unit sphere, and it is known that maximal perimeter polyhedra
inscribed in the sphere must be regular.  So for a 4-clique, the vector labels
must form a regular tetrahedron.  We carefully argue that for regular
tetrahedra, fixing two of the vertices (approximately) fixes the other two vertices up to
swapping (several of these geometric facts are proved in
\cref{sec:geometryLemmas}).
This means that the clique gadget corresponding to an edge $uv$ in $G$ shares
two vertices with any gadget corresponding to an edge $vw$ incident to the
first edge: $v$ and the ``global'' ancilla shared by all gadgets. This means
that, once we fix the vector assigned to the global ancilla, the choice of a
vector for $v$ restricts the labels of both $u$ and $w$ to be chosen from a set
of two vectors. So simultaneously optimizing every clique gadget is possible
iff, for each connected component of $G'$, we can 3-color that component with
three vectors (corresponding to, for any $v$ in the component, the vector
assigned to $v$, the vector assigned to the global ancilla vertex, and the two vectors that can share a maximal
tetrahedron with those two).

Finally, when the two largest entries of $W$ are equal but distinct from the
third, we combine the two previous approaches.  Inserting star gadgets
effectively reduces the problem from three dimensions to two, by penalizing
vector assignments not in the 2d space corresponding to the two largest entries
of $W$.  We then add $3$-clique gadgets, with one ancilla for each edge in $G$,
and optimizing these over two dimensions corresponds to inscribing maximal
perimeter triangles in the unit circle. Now, assigning a vector to one vertex
fixes the optimal vectors assigned to the other two vertices (again up to swapping them) and
so there is again a one-to-one correspondence between vector assignments
simultaneously optimizing every clique gadget and 3-colorings of the
connected components of $G$.

\paragraph{Three-Dimensional Vector Max-Cut} Our proof that $\kmaxcut[k]$ with
positive unit weights is $\NP$-complete runs on very similar lines to our
\threeColoring{} reduction for $\Wlinearmaxcut$ with all weights equal.
However, there is one additional complication: that proof depended on the fact
that the sum of the side lengths of the tetrahedron is uniquely maximized by
choosing it to be regular. This is not the case when we instead consider the
squared side lengths, for instance assigning half the vectors to one pole of
the sphere and half to the other would achieve the same bound. So we use a
different gadget: instead of replacing every edge with a 4-clique, we
replace every edge with a 4-clique that in turn has its edges replaced with
triangles. It turns out that this gadget \emph{does} have a unique optimal
$\kmaxcut[k]$ assignment, which in particular assigns the vectors in the
4-clique to a regular tetrahedron, allowing us to proceed along the same
lines as the aforementioned proof.

\subsection{Related Work}\label{ssec:relatedWork}
Brand{\~a}o and Harrow~\cite{harrowbrandao2016product} give simple conditions under
which $2$-local Hamiltonians have product states
achieving near-optimal energy, such as systems with high-degree interaction graphs.
This suggests that unless $\NP = \QMA$,
such Hamiltonians cannot be $\QMA$-complete.
Since then, there has been a line of work on the relationship between
product states and general ground states in other Hamiltonians,
in the more general case when the two problems are not equivalent.
See \eg{} \cite{bravyi2019approximation,gharibian2012approximation,hallgren2020approximation}.

Product states have especially been studied in the context of the \qmcLong{} problem, introduced by
Gharibian and Parekh \cite{gharibian_et_al:LIPIcs:2019:11246}.
Bri{\"e}t, de Oliveira Filho, and Vallentin~\cite{briet2010positive}
give hardness results conditional on the
unique games conjecture for approximating the optimal product state in the \qmc{} model.
Related work by Hwang, Neeman, Parekh, Thompson, and
Wright~\cite{hwang2023unique} also gives tight hardness results for the \qmc{}
problem under a plausible conjecture.
Parekh and Thompson~\cite{parekh2022optimal} give an optimal approximation algorithm for
\qmc{} when using product states.

See \cite{hwang2023unique} for an exposition of the relationship between
\qmc{} restricted to product state solutions and the \prob{Vector Max-Cut} problem.
Studying vector solutions to \maxcut{} has a long history \cite{lovasz2019graphs},
including the seminal Goemans-Williamson algorithm \cite{goemans1995improved}.
This study is usually with the goal of a solution to the original \maxcut{} problem,
which relates to approximation ratios of integer and semidefinite programs.
Bounding these ratios has been referred to in terms of Grothendieck problems and inequalities: see \cite{briet2011generalized,bandeira2016approximating} for further context on this nomenclature.
A tight NP-hardness result is known for the non-commutative Grothendieck problem
\cite{briet2015tight} which also generalizes the ``little'' Grothendieck
problem over orthogonal groups \cite{bandeira2016approximating}.
Iterative algorithms (heuristics) for solving \kmaxcut{} are also well-studied in the
literature (see \cite{burer2003nonlinear} and citing references), since in
practice solving \kmaxcut{} is often faster than solving the
corresponding SDP relaxation.

Cubitt and Montanaro \cite{CM16_hamiltonians} classified \sLH for sets $\calS$
of 2-qubit terms.  This work relies in turn on the work of Bravyi and Hastings
\cite{bravyi2017complexity} to classify the $\StoqMA$ case.
\cite{CM16_hamiltonians} also examined a variant where $\calS$ is assumed to
always contain arbitrary single-qubit terms.  Follow-up work by Piddock and
Montanaro \cite{PM17_hamiltonians} began investigating classifying the
complexity of \sLH under the additional restrictions of positive weights
(anti-ferromagnetic model) and/or interactions restricted to a spatial geometry
such as a 2D grid.  \cite{Cubitt_2018,piddock2018universal} continued along
these lines, and introduced Hamiltonian simulations rather than computational
reductions.

In classical computer science,
Schaefer \cite{schaefer1978complexity} gave a dichotomy theorem showing that given a
fixed set of allowed constraints, the family of CSPs is either decidable in $\PTIME$ or is
$\NP$-complete; but see \cref{sec:prelim} or \cite{CM16_hamiltonians} for some important
assumptions that are made in the quantum versus classical models.
In fact, Schaefer's classification offers a more fine-grained classification with classes
within $\PTIME$.
Later, \cite{creignou1995dichotomy,khanna1997complete} gave a similar dichotomy
theorem for the complexity of \prob{Max-SAT} and related optimization problems,
where the question is not just whether all
clauses are simultaneously satisfiable but how many are simultaneously satisfiable.
Applying weights to constraints becomes relevant with \prob{Max-SAT}, and is covered
in their work.
This work is especially relevant given \kLH is more analogous to \prob{Max-SAT} than to \prob{SAT}.
Continuing, Jonsson \cite{jonsson2000boolean} classified these problems
when both positive and negative weights are allowed. While arbitrary weights seem natural
in the quantum setting, previous classical work simply assumed all weights
were non-negative.
The book \cite{creignou2001complexity} offers an excellent survey of this area.
The more recent results of \cite{jonsson2006approximability,thapper2016complexity}
extend classical classification results to problems with non-binary variables
(analogous to qu\emph{d}its).

\subsection{Open Questions}
\begin{enumerate}[itemsep=0mm]
\item We have shown a relationship between when product state problems and general Hamiltonian
problems are hard. This points toward an important question: can some ``interesting'' class of local Hamiltonians or a Hamiltonian problem for which we do not know an explicit efficient algorithm be proven \emph{not} hard, \eg{} neither $\NP$- nor $\QMA$-hard, by showing the corresponding product state problem is in $\PTIME$?

\item Is a more refined classification of the complexity of product state
problems, taking into account allowable weights or spatial geometry in the vein
of~\cite{PM17_hamiltonians}, or imposing other promises, possible?

\item While little progress has been made classifying the general $\sLH$ problem for higher-locality
families, can we classify \sprodLH for families of $k$-local Hamiltonians with $k>2$?

\item Can we relate approximability instead of just complexity? For example,
does the ability to $\alpha$-approximate the product state problem imply the
ability to $\beta$-approximate the general ground state problem on families
defined by some sets of allowed interactions but not others?

\item As mentioned above, we make the first progress towards a conjecture of Lov{\'a}sz \cite{lovasz2019graphs} that \kmaxcut{} is $\NP$-hard for any $k=O(1)$.
We only focus on $k=3$ here because of our interest in \prodLH{}.
Can our proof be generalized to other values of $k$?

\end{enumerate}

\section{Preliminaries}\label{sec:prelim}
We assume familiarity with the conventions of quantum computation
\cite{watrous2018theory} and complexity theory
\cite{arora2009computational,kitaev2002classical}.
See also \cite{gharibian2015quantum,CM16_hamiltonians} for surveys of Hamiltonian complexity.

\subsection{Notation}
$I$ denotes the identity operator.  $\operatorname{\lambda_{\min}}(H)$ and
$\operatorname{\lambda_{\max}}(H)$ denote the minimum and maximum eigenvalues
of an operator $H$.  In the same manner as with asymptotic $\bO{\cdot}$
notation, we use $\poly(n)$ to denote a term that can be bounded by some fixed
polynomial in $n$.

For an operator $A$, we use superscripts such as $A^{abc}$ to indicate $A$ acts on individual qubits $a,b,$ and $c$. Unless $A$ is symmetric, the order matters and $A^{ab}$ is different than $A^{ba}$.
If no superscripts are used, then the action is implicit in the ordering of the terms (left versus right).

When clear from context, we will denote the tensor product of two operators $A\otimes B$ simply by $AB$.
All terms implicitly are tensor the identity on any systems not specified.

$\SWAP$ will denote the 2-qubit operator exchanging $\ket{01}$ and $\ket{10}$
while $\ket{00}$ and $\ket{11}$ unchanged. We call a 2-qubit term $H$
\emph{symmetric} if $H=\SWAP (H) \SWAP$, meaning the ordering of the qubits
does not matter.  Alternatively, $H$ is \emph{antisymmetric} if $H = - \SWAP
(H) \SWAP$.

The single-qubit Pauli matrices are denoted $X,Y,Z$ or $\sigma_1,\sigma_2,\sigma_3$. Recall that $\{X,Y,Z,I\}$ is a basis for $2\times 2$ Hermitian matrices.
The \emph{Pauli decomposition} of a 2-qubit Hermitian matrix $H$ is $H$ written in the Pauli basis,
\begin{equation}\label{eqn:pauliDecomp}
	H = \sum_{i,j=1}^3 M_{ij} \sigma_i \sigma_j + \sum_{k=1}^3 \paren{v_k \sigma_k I + w_k I\sigma_k} ,
\end{equation}
with all coefficients real and the $3\times 3$ matrix $M$ referred to as the \emph{correlation matrix}.
Generally, \cref{eqn:pauliDecomp} should include a term $wII$, but we will work with traceless terms
such that $w=0$.

Unless otherwise stated, all graphs are undirected and simple, meaning there are no self-loops and no multi-edges.
We assume all graphs are connected, as it is straightforward to extend any of our
constructions to disconnected graphs.
When summing over edges, $\sum_{ij\in E}$, we do not double-count $ij$ and $ji$.
Finally, $S^{i} = \{x\in \reals^{i+1} : \norm{x}=1\}$ denotes the unit sphere in $(i+1)$-dimensional space.

\subsection{Definitions and Assumptions}

A $k$-local Hamiltonian on $n$ qubits is a Hermitian matrix $H \in
\reals^{2^n\times 2^n}$ that can be written as $H =\sum_{i=1}^m H_i$ such that
each $H_i$ is Hermitian and acts non-trivially on at most $k$ qubits.  More
precisely, each $H_i$ acts on some subset $S_i$ of at most $k$ qubits and each
term in the sum is $H_i \otimes I^{[n]\setminus S_i}$, but we generally leave
this implicit.  We usually consider constant values of $k$, so each term is of
constant size independent of $n$.  The $k$-qubit terms $H_i$ are often referred
to as \emph{interactions} between qubits.  We may refer to eigenvalues and
expectation values $\bra{\psi}H\ket{\psi}$ as the \emph{energy} of the state
$\ket{\psi}$ in the system described by $H$.  In particular, the \emph{ground
state energy} and \emph{ground state} refer to the minimum eigenvalue and an
associated eigenvector.

Estimating the minimum eigenvalue of a Hamiltonian is a natural quantum generalization of estimating the maximum number of satisfiable clauses in a Boolean formula.

\begin{definition}[\kLH]\label{def:kLH}
	Given a $k$-local Hamiltonian $H=\sum_{i=1}^m H_i$ acting on $n$ qubits with $m=\poly(n)$, the entries of each $H_i$ specified by at most $\poly(n)$ bits, and the norms $\norm{H_i}$ polynomially bounded in $n$,
	and two real parameters $b,a$ such that $b-a\geq 1/\poly(n)$, decide whether $\lambdamin(H)$ is at most $a$ (YES) or at least $b$ (NO), promised that one is the case.
\end{definition}

In this work, we are interested in $\kLH$ restricted to families of local Hamiltonians, where the families are determined by sets of allowed interactions.
In particular, we will be interested in sets of 2-qubit interactions.

\begin{definition}[\sLH]\label{def:fLH}
	For $\calS$ any fixed set of Hamiltonian terms,
	define \sLH{} as the problem $\kLH$ with the additional promise that any input is of the form $\sum_{i=1}^m w_i H_i$ where each $H_i$ is an element of $\calS$ assigned to act on some subset of qubits and the \emph{weights} $w_i \in \reals$ have magnitude polynomially-bounded in $n$.
\end{definition}

\begin{remark}\label{remark:sLHassumptions}
There are several standard assumptions implicit in our definition of \sLH{}.
Some are not physically realistic in the context of the condensed-matter literature
but allow us to precisely characterize the complexity of these problems.
First, although classical CSPs generally allow a constraint to take as input
multiple copies of the same variable, this makes less sense in the quantum
setting and we do not allow it.  Second, the definition of $k$-local only
restricts the dimension of each term, it does not imply any spatial locality or
geometry.  Therefore, any term in $H$ may be applied to any subset of qubits
with the qubits arranged in any order.  In particular this means that, if
$\calS$ contains a directed term $H^{ab}$, then the family of Hamiltonians
allowed as input to $\sLH$ is equivalent to the family allowed to
$\sLH[\calS']$ for $\calS'=\calS\cup\{H^{ba}\}$.  Third, for the purpose of
classifying the complexity of \sLH{}, we may assume $I\in \calS$, since adding
or removing a term $wI$ is equivalent to simply shifting the input parameters
$a,b$ by $w$.  For $\calS$ containing 2-qubit terms, this fact also implies we
may assume all terms in $\calS$ are traceless.  Fourth, except when noted,
we allow both positive and negative weights.
\end{remark}

Classifying the complexity of systems under additional, more
physically natural restrictions appears to be a significantly more difficult
problem \cite{CM16_hamiltonians,PM17_hamiltonians}.

Given this setup, our interest will be in the problems $\kLH$ and $\sLH$ restricted to product states.

\begin{definition}[Product state]\label{def:productState}
  A state $\rho=\bigotimes_{i=1}^n \rho^i$ where each $\rho^i$ is a single-qubit state.
\end{definition}

\begin{definition}[$\prodLH$]\label{def:prodLH}
  Given a $k$-local Hamiltonian $H=\sum_{i=1}^{m} H_i$ on $n$ qubits with
  $m=\poly(n)$, the entries of each $H_i$ specified by at most $\poly(n)$ bits, and the norms $\norm{H_i}$ polynomially bounded in $n$, and two real parameters $b\geq a$, decide whether there exists a product state $\rho$ with
  $\tr(\rho H)\leq a$ (YES) or all product states satisfy
  $\tr(\rho H) \geq b$ (NO), promised that one is the case.
\end{definition}

The problem $\sprodLH$ is defined analogously.
In both definitions, the fact that product states have concise classical descriptions allows us to naturally consider any choice of parameters, even an exact decision problem with $b= a$, in contrast to \kLH.

By convexity, a product state $\rho$ achieves an extreme value of $\tr(\rho H)$ if and only if there exists a pure product state $\ket{\psi}$ which achieves that value.
Similarly, mixtures of product states, known as separable states, reduce to considering pure product states.

\begin{remark}\label{remark:Scontains}
  In the context of $\sLH$ or $\sprodLH$ given some fixed set $S$, we will
describe $\calS$ as ``containing'' a Hamiltonian term $H$, or that we ``have
access to'' $H$, even when formally $H\notin \calS$.  As previously referenced
in \cref{remark:sLHassumptions}, given a set $\calS$, the family of
Hamiltonians allowed as input to $\sLH$ may be equivalent to the family allowed
given some other set $\calS'$.  For example, $\calS'$ may include
$\{PHP^{\dagger}\}$ for $H\in \calS$ and any permutation of the qubits $P$.
Similarly, adding constant multiples of the terms in $\calS$ or any linear
combinations of terms from $\calS$ does not change the family of allowed
Hamiltonians.  So, when discussing $\sLH$, we may implicitly refer to elements
of the largest set $\calS'$	such that $\sLH$ and $\sLH[\calS']$ each have the
same family of allowed inputs.

Additionally, we note that \sLH{} is reducible to $\sLH[\calS']$ for any
$\calS'$ which can be used to implement all elements of $\calS$---whether
because formally $\calS\subseteq \calS'$ or through other means.  In the
opposite direction, if the terms in a set $\calS$ can be used to construct some
term $H$ and we wish to show hardness of $\sLH$, then it is sufficient to show
hardness of $\sLH[\{H\}]$.
\end{remark}

Finally, for a 2-local Hamiltonian, we may refer to the \emph{interaction
graph}, with vertices associated with each qubit such that vertex $i$ is
adjacent to vertex $j$ whenever a nonzero interaction exists on qubits $i$ and
$j$.  When all interactions are symmetric, then the graph is undirected.
Notably, when $\sLH$ is defined with $\calS$ a singleton, then an input is
fully specified by its weighted interaction graph.

\section{\texorpdfstring{Classification of \sprodLH{}}{Classification of S-ProdLH}}
\label{sec:2localHard}

In this section, we prove a dichotomy theorem classifying the complexity of
estimating the minimum expectation of product states for given families of
2-local Hamiltonians.  In particular, we show that for any set $\calS$ of
2-qubit terms such that at least one term is not 1-local,\footnote{We would
prefer a more concise name for 2-qubit terms that are not 1-local, but are
unaware of any. One option is 2-qubit terms with \emph{Pauli degree} 2.
Alternatively, these are 2-qubit terms which have nonzero \emph{Pauli rank},
referring to the rank of the correlation matrix $M$ in the Pauli decomposition
(\cref{eqn:pauliDecomp}).  } the problem $\sprodLH$ is $\NP$-complete.  These
are precisely those 2-local families such that (as shown
in~\cite{CM16_hamiltonians}) $\sLH$ is $\QMA$-, $\StoqMA$-, or $\NP$-complete.
Conversely, if all terms are 1-local, then both $\sLH$ and $\sprodLH$ are in
$\PTIME$.

\prodStatesHard*

\noindent Our $\NP$-hardness results hold with coupling strengths of
at most constant magnitude.

For comparison with our classification, we recall the tetrachotomy theorem of
Cubitt and Montanaro \cite{CM16_hamiltonians} classifying \sLH{} for families
of 2-local Hamiltonians.  They proved that for every set of 2-qubit Hamiltonian
terms $\calS$, the problem $\sLH$ is either in $\PTIME$ or $\NP$-, $\StoqMA$-,
or $\QMA$-complete, and described properties of the set $\calS$ which determine
the problem's complexity.  We note that both \cref{thm:CM16} and our
\cref{thm:prodStatesHard} classify the complexity of \emph{all} sets of 2-qubit
terms.

\begin{theorem}[Theorem 7 of \cite{CM16_hamiltonians}]\label{thm:CM16}
	For $\calS$ any fixed set of 2-qubit Hamiltonian terms:
	\begin{itemize}
		\item If every matrix in $\calS$ is 1-local, then \sLH{} is in $\PTIME$.
		\item Otherwise, if there exists a single-qubit unitary $U$ such that $U$
		locally diagonalizes all elements of $\calS$ (i.e.\ $U^{\otimes 2} H
		U^{\dagger \otimes 2}$ is diagonal for each 2-qubit $H \in \calS$), then
		\sLH{} is $\NP$-complete.  \item Otherwise, if there exists a single-qubit
		unitary $U$ such that for each 2-qubit $H \in\calS$, \[
		U^{\otimes 2}H U^{\dagger \otimes 2} = \alpha Z^{\otimes 2}+A \otimes I + I\otimes B
		\]
		for some $\alpha \in \reals$ and 1-local Hermitian matrices $A,B$, then
		\sLH{} is $\class{StoqMA}$-complete.
		\item Otherwise, \sLH{} is $\QMA$-complete.
	\end{itemize}
\end{theorem}

Combining our classification of \sprodLH{} with the classification of \sLH{} gives us \cref{corr:fLHeasyProdLHeasy,corr:fLHeasyProdLHhard}.

To prove \cref{thm:prodStatesHard}, showing containment in $\PTIME$ and $\NP$ are straightforward, and our effort is to prove $\NP$-hardness.
In the proof, we will use a simple symmetrization trick that allows us to consider only antisymmetric or symmetric Hamiltonian terms.
We then prove two lemmas, one for each case.

\begin{lemma}\label{lem:antisymProdLH_posProdLH}
	If $\calS$ contains a 2-qubit antisymmetric term that is not 1-local,
	then $\Wlinearmaxcut$ with $W=\diag(1,1,0)$
	is polynomial-time reducible to $\sprodLH$.
\end{lemma}

\begin{lemma}\label{lem:symmetricProdLH}
	If $\calS$ contains a 2-qubit symmetric term that is not 1-local,
	then there exists a fixed non-negative $W=\diag(\alpha,\beta,\gamma)$ with at least one of $\alpha,\beta,\gamma$ nonzero
	such that $\Wlinearmaxcut$ is polynomial-time reducible to $\sprodLH$.
\end{lemma}
In \cref{sec:vectorProblems} we prove \cref{thm:WlinMaxCut}, that $\Wlinearmaxcut$ is $\NP$-complete for any $W=\diag(\alpha,\beta,\gamma)$ that is nonzero and non-negative.

We state some helpful facts in \cref{ssec:closurePropsOfsLH} below and then
prove the two lemmas in \cref{ssec:proofsOfLemmas}.  We will now use these
lemmas to prove our main theorem.

\begin{proof}[Proof of \cref{thm:prodStatesHard}]
	First consider the case where $\calS$ only contains 1-local terms. Then we
	can write $H = \sum_i H_i$, where $H_i$ acts only on the $i\nth$ qubit. If
	$\ket{\psi_i}$ is the single-qubit state minimizing $\braket{\psi_i | H_i | \psi_i}$,
	then $\ket{\psi} = \bigotimes_i \ket{\psi_i}$ minimizes $\braket{\psi | H | \psi}$,
	and so $\sprodLH$ can be solved by finding the ground state of $n$ single-qubit
	Hamiltonians, which takes $\bO{n}$ time.

    Now suppose $\calS$ is a set of of 2-qubit Hamiltonian terms such that at
    least one element of $\calS$ is not 1-local. Let $H$ be any such element.
    As previously mentioned, $\sprodLH \in \NP$ for any fixed $\calS$, as product states have concise classical descriptions which can be used to efficiently calculate expectation values for a given local Hamiltonian.
    If $H$ is antisymmetric, then \Wlinearmaxcut{} with fixed $W=\diag(1,1,0)$ is reducible to \sprodLH{} by \cref{lem:antisymProdLH_posProdLH}, and \Wlinearmaxcut{} with such a $W$ is $\NP$-hard by \cref{thm:WlinMaxCut},
    so \sprodLH{} is $\NP$-complete.
    If $H$ is symmetric, then by \cref{lem:symmetricProdLH} there exists a non-negative nonzero matrix $W=\diag(\alpha,\beta,\gamma)$ such that \Wlinearmaxcut{} is reducible to \sprodLH{},
    it is $\NP$-hard by \cref{thm:WlinMaxCut},
    and so \sprodLH{} is $\NP$-complete.
    If $H$ is neither of these, then we use our freedom to permute the direction $H$ is applied to any pair
    of qubits $a,b$ to apply both $H^{ab}$ and $H^{ba}$, which is equivalent to implementing the
    symmetric term $H' = H +\SWAP{} (H) \SWAP$.
    So, we can say $\calS$ effectively contains $H'$, or formally that hardness of $\sprodLH[\{H'\}]$ implies hardness of $\sprodLH$, and again referring to \cref{lem:symmetricProdLH} concludes the proof.
\end{proof}

\subsection{\texorpdfstring{Closure Properties of $\sprodLH$}{Closure Properties of S-ProdLH}}
\label{ssec:closurePropsOfsLH}

Before proving the two lemmas required in the proof of \cref{thm:prodStatesHard}, we review
several more facts regarding 2-qubit Hamiltonian terms and operations under which
the complexities of $\sLH$ and $\sprodLH$ are unaffected.
This section mostly reviews observations made in \cite{CM16_hamiltonians}.

First, for a single-qubit unitary $U$ and an operator $H$,
define \emph{simultaneous conjugation} by $U$ to mean $U^{\otimes n} H U^{\dagger \otimes n}$.
When discussing sets $\calS$ of $k$-qubit terms, we define simultaneous conjugation
to mean $\{U^{\otimes k} H U^{\dagger\otimes k} | H\in \calS \}$.

\begin{fact}\label{fact:onelocalconjugating}
  For any single-qubit unitary $U$, the complexities of \sLH{} and \sprodLH{} are
  equal to the complexities of \sLH[\calS'] and \sprodLH[\calS'], respectively,
  where $\calS'$ is $\calS$ simultaneously conjugated by $U$.
\end{fact}

\noindent
Observe that $U^{\otimes n}\left(\sum_{i=1}^m H_i\right) U^{\dagger\otimes n} = \sum_{i=1}^m U^{\otimes k} H_i U^{\dagger \otimes k}$.
Simultaneous conjugation by $U$ gives
a bijection between Hamiltonians allowed in \sLH{} and \sLH[\calS'] as well
as \sprodLH{} and \sprodLH[\calS'].
The above fact follows from observing that this bijection preserves
expectation values, and that $U^{\otimes n}\ket{\phi}$ is a product state iff
$\ket{\phi}$ is.

As an application of \cref{fact:onelocalconjugating}, the following is based on an observation in \cite{CM16_hamiltonians}.

\begin{fact}\label{fact:relabelPaulis}
  For any choice of permutation $\pi$ on $\{1,2,3\}$ and any choice of \emph{two} of $c_1,c_2,c_3=\pm 1$,
  there exists a single-qubit unitary $U$ and corresponding third coefficient such that
  simultaneous conjugation by $U$ maps
  the Pauli matrices $\{\sigma_1,\sigma_2,\sigma_3\}$ to $\{c_{\pi(1)}\sigma_{\pi(1)},c_{\pi(2)}\sigma_{\pi(2)},c_{\pi(3)}\sigma_{\pi(3)}\}$.
  So, writing every element of $\calS$ in the Pauli basis,
  relabeling all $\sigma_i$ with $c_{\pi(i)} \sigma_{\pi(i)}$ in the decompositions of each
  element of $\calS$ does not change the complexity of $\sLH$ or $\sprodLH$, where $\pi$ and two of the coefficients can be chosen arbitrarily.
\end{fact}

To justify the above fact, consider simultaneously rotating the three axes of the Bloch sphere.
Next, we quote the following, more involved, fact without proof.

\begin{fact}[\cite{CM16_hamiltonians, horodecki1996information}]\label{fact:local_rot}
  Let $H$ be a 2-qubit Hamiltonian term with Pauli decomposition
  $H=\sum_{i,j=1}^{3} M_{ij} \sigma_i \sigma_j + \sum_{k=1}^{3} \paren{v_k \sigma_k I
  + w_k I \sigma_k}$. For any single-qubit unitary $U$,
	\begin{equation}\label{eq:rot_pauli_decomp}
    U^{\otimes 2} H (U^\dagger)^{\otimes 2}=\sum_{i,j=1}^{3} (R M R^T)_{ij}
    \sigma_i  \sigma_j + \sum_{k=1}^{3}\paren{(Rv)_k \sigma_k I + (Rw)_k I
    \sigma_k}
	\end{equation}
  for some $R\in SO(3)$.  Likewise, for any $R\in SO(3)$, there exists a
  single-qubit $U$ such that the Pauli decomposition of $U^{\otimes 2} H
  (U^\dagger)^{\otimes 2}$ matches \cref{eq:rot_pauli_decomp}.
\end{fact}

A further straightforward observation from \cite{CM16_hamiltonians} is that in the Pauli decomposition (\cref{eqn:pauliDecomp}), if $H$ is symmetric then the correlation matrix $M$ is symmetric,
and if $H$ is antisymmetric then $M$ is skewsymmetric, meaning $M=-M^{\top}$.

Finally, the below observation combines some of the above facts to establish a ``normal form'' for symmetric and antisymmetric terms.

\begin{fact}\label{remark:symAntiSym}
	If a 2-qubit Hamiltonian term $H$ is symmetric, and so the associated correlation matrix $M$ is symmetric, there exists $R\in SO(3)$ which diagonalizes $M$.
	Combining \cref{fact:local_rot,fact:onelocalconjugating},
	for any 2-qubit symmetric term $H$,
	there exists a symmetric term of the form $H'=\sum_{i=1}^3 u_i \sigma_i \sigma_i + \sum_{j=1}^3 v_j (\sigma_j I + I \sigma_j)$ such that
	the complexities of
	$\sLH[\{H\}]$ and $\sprodLH[\{H\}]$ are respectively the same as
	$\sLH[\{H'\}]$ and $\sprodLH[\{H'\}]$.

	If $H$ is a 2-qubit antisymmetric term that is not 1-local, then $M$ is skewsymmetric and nonzero.
	Such an $M$ may be block diagonalized via some $R\in SO(3)$ such that $H$ is mapped to $a\!\left(\sigma_i\sigma_j - \sigma_j\sigma_i\right) + \sum_{k=1}^3 v_k \left(\sigma_k I - I\sigma_k\right)$ \cite{thompson1988normalForm,CM16_hamiltonians}.
        In particular, by \cref{fact:relabelPaulis}, $H$ can be mapped to $a\!\left(XZ-ZX\right) + \sum_{k=1}^3 v'_k \left(\sigma_k I - I\sigma_k\right)$.
	Therefore, the complexities of $\sLH[\{H\}]$ and $\sprodLH[\{H\}]$ are unaffected by assuming $H$ has this form.
\end{fact}

\subsection{Proofs of Antisymmetric and Symmetric Lemmas}\label{ssec:proofsOfLemmas}

We now prove the two lemmas required in the proof of the main theorem, respectively handling the cases that $\calS$ contains an antisymmetric term and that $\calS$ contains a symmetric term.
In both cases, it is sufficient for $\calS$ to contain just a single term.
Interestingly, our construction in \cref{lem:antisymProdLH_posProdLH} for antisymmetric terms is unweighted, meaning all weights are $+1$. In this case, the final Hamiltonian is fully determined by the specification of a single 2-qubit term and the interaction graph.
Our construction in \cref{lem:symmetricProdLH} uses positive and negative unit weights, $\pm 1$.
Intuitively, antisymmetric terms inherently allow negativity by simply permuting the qubits they act on, while for symmetric terms we must use negative weights.

\begin{proof}[Proof of \cref{lem:antisymProdLH_posProdLH}]
	Consider an arbitrary instance of the problem \Wlinearmaxcut{} with $W=\diag(1,1,0)$.
	For a given graph $G=(V,E)$, the objective function is
	\[
		\Wlinearmaxcutop(G) = \frac{1}{2}~\max_{\hat{\imath} \in S^{2}} \sum_{ij\in E} \norm{W\hat{\imath}-W\hat{\jmath}} .
	\]
	Given a parameter $C$, we must decide whether $\Wlinearmaxcutop(G)$ is at least $C$ or at most $C-\eps$.
	To reduce \Wlinearmaxcut{} to \sprodLH{}, we first construct a gadget using the promised antisymmetric term.
	Then, we apply this gadget according to the graph $G$ such that the minimum energy of a product state in our final Hamiltonian will equal $-\Wlinearmaxcutop(G)$.

	Denote the assumed 2-qubit antisymmetric term that is not 1-local in $\calS$ by $H$.
	By \cref{fact:local_rot}, our antisymmetric term $H$ may be mapped to a term of the form
    \[
        w \left(X^aZ^b-Z^aX^b\right) + \sum_{k=1}^3 v_k \left(\sigma_k^a I^b - I^a\sigma_k^b\right)
    \]
	where all coefficients are real, and we have $w \neq 0$ since the term is
not 1-local.  As explained in \cref{remark:Scontains} and
\cref{ssec:closurePropsOfsLH}, the complexity of $\sprodLH[\{H\}]$ is
equivalent to that of $\sprodLH[\{H'\}]$ for $H'$ derived using a variety of
operations, including permutations and linear combinations.  If $w$ is
negative, then we redefine the direction the term acts in, $H^{ab}$ versus
$H^{ba}$, so that $w$ is positive.  Finally, we scale\footnote{If we want the
terms to have unit weights, we could forgo scaling the term and reduce to
$w\Wlinearmaxcutop$ instead. As $w > 0$ this problem has the same complexity
as $\Wlinearmaxcutop$.} the term so that $w=1$ and define a single-qubit
Hermitian matrix $A=\sum_{k=1}^3 v_k \sigma_k$.  Given the complexity is
unchanged using $H$ or $H'$, we simply redefine the original term, so that
    \[
        H^{ab} = X^aZ^b - Z^aX^b + A^aI^b - I^aA^b .
    \]

    Next, we use a symmetrization gadget to remove the 1-local terms $AI-IA$.
    For four qubits $a,b,c,d$, define
    \[
        B = H^{ab} + H^{bc}+H^{cd}+H^{da} .
    \]
    Note that here the direction of the interaction matters, since the terms are asymmetric.
    Then
    \[
        B = (X^a - X^c)(Z^b-Z^d) - (Z^a-Z^c)(X^b-X^d) .
    \]

    Now we consider how $B$ interacts with product states on four qubits.
    For $e=a,b,c,d$, define
    \[
	    \rho^e = \frac{1}{2}(I + r^e\cdot v^e)
	\]
	with $v^e=(X^e,Y^e,Z^e)$ the Pauli operators and $r^e = (x^e,y^e,z^e)$ the Bloch vector.
	Then, writing any product state on qubits $a,b,c,d$ as $\rho^a\rho^b\rho^c\rho^d$,
    the expectation value on $B$ is $\tr\left(B \rho^a \rho^b \rho^c \rho^d\right)$.
    After eliminating terms, we find this equals
    \begin{gather*}
        (r^a - r^c)^{\top} W'  (r^b-r^d) \quad\text{for}\quad W'=\begin{bmatrix} 0 & 0 & 1\\ 0 & 0 & 0 \\ -1 & 0 & 0\end{bmatrix} .
    \end{gather*}
    It is helpful to note that $W'^{\top} = -W'$.

    Now, consider the graph $G$ given as input to \Wlinearmaxcut.
    Associate a qubit with each vertex and call these the ``vertex qubits''.
    For each edge $ij$, construct a copy of $B$ such that it acts on qubits $ij$ and two ancilla qubits.
    The vertex qubits may be shared among several gadgets, while the ancilla qubits are part of only one gadget.
    In particular, we choose to associate the vertex qubits with qubits $a$ and $c$ in each copy of $B$, letting $b$ and $d$ be the ancilla.
    We will refer to the copy of $B$ which acts on vertex qubits $i$ and $j$ as $B^{ij}$.
	Our Hamiltonian is then
	\[
		H_{\text{final}} = \sum_{ij\in E} B^{ij} .
	\]

	Before analyzing the full Hamiltonian $H_{\text{final}}$,
    consider the minimum expectation of a single gadget $B^{ij}$ if the two vertex qubits are fixed, i.e.\ $\min_{r^b,r^d} (r^i - r^j)^{\top} W' (r^b-r^d)$.
    The minimum is achieved when $r^b = -W'^{\top}(r^i-r^j) / \norm{W'^{\top}(r^i-r^j)}$ and
    $r^d=-r^b$, which yields an expectation of
    \begin{gather*}
        - 2 \norm{W''(r^i-r^j)} \quad\text{for}\quad W''=\diag(1,0,1).
    \end{gather*}
    Therefore, given an arbitrary graph $G=(V,E)$,
    applying our gadget to every edge constructs a Hamiltonian
    such that the minimum expectation of any product state is equal to
    \begin{gather*}
        2\min_{r^i \in S^2} \sum_{ij\in E}
        -\norm{W''(r^i-r^j)}
        = -2\max_{r^i \in S^2}
        \sum_{ij\in E} \norm{W''(r^i-r^j)} .
    \end{gather*}

	For any graph $G$, the objective $\Wlinearmaxcutop(G)$ is equal for $W=\diag(1,1,0)$ and $W''=\diag(1,0,1)$.
	Alternatively, we may use our freedom to relabel Paulis to redefine the 2-qubit term $H$ such that the final weight matrix is $W$.

	Finally, multiplying the full Hamiltonian $H_{\text{final}}$ by $\frac{1}{2}$ gives us that the minimum expectation of any product state equals
    \[
        -\max_{r^i \in S^2} \sum_{ij\in E} \norm{W(r^i-r^j)} = -\Wlinearmaxcutop(G).
    \]

    We conclude deciding \Wlinearmaxcut{} reduces to deciding \prodLH{} on $H_{\text{final}}$.
    Since $H_{\text{final}}$ is entirely constructed from the antisymmetric term $H\in\calS$, this completes the desired reduction of \Wlinearmaxcut{} with $W=\diag(1,1,0)$ to \sprodLH{}.
\end{proof}

Next we prove the lemma dealing with $\calS$ containing a symmetric term. The construction is nearly the same as the in the previous proof, but requires negative weights to implement the symmetrization gadget removing 1-local terms.

\begin{proof}[Proof of \cref{lem:symmetricProdLH}]
	Given fixed $\calS$, we will show there exists some fixed $W$ such that \Wlinearmaxcut{} reduces to \sLH{}.
	But, before describing $W$, we must analyze $\calS$.

	Denote the assumed 2-qubit symmetric term that is not 1-local in $\calS$ by $H$.
	As in the previous proof, without changing the complexity of $\sprodLH[\{H\}]$ we may
	conjugate and scale as necessary so that
    \[
        H^{ab} = \alpha^- X^aX^b + \beta^- Y^aY^b + \gamma^- Z^aZ^b +  \sum_{j=1}^3 v_j (\sigma_j^a I^b + I^a \sigma_j^b) ,
    \]
    where all coefficients are real and at least one of $\alpha^-,\beta^-,\gamma^-$ is nonzero since $H$ is nonzero.
    The superscripts in the above equations are to differentiate the coefficients of $H$ from the entries of $W$, which must be non-negative.

    We again use a symmetrization gadget to remove the 1-local terms, but now require negative weights.
    Given four qubits $a,b,c,d$,
    define $B= H^{ab} + H^{cd} - H^{bd} - H^{ac} $.
    This is a rectangle with two positive edges and two negative edges. Then
    \[
        B = \alpha(X^a - X^d)(X^b - X^c) + \beta (Y^a - Y^d)(Y^b - Y^c) + \gamma (Z^a
        - Z^d)(Z^b - Z^c) .
    \]
    Now we see how $B$ interacts with product states on four qubits.
    For $e = a,b,c,d$, we again define $\rho^e = \frac{1}{2}\paren{I + r^e \cdot v^e}$
    with $v^e = (X^e, Y^e, Z^e)$ and $r^e = (x^e, y^e, z^e)$. Then, writing any
    product state on $a,b,c,d$ as $\rho^a\rho^b\rho^c\rho^d$, the expectation
    value on $B$ is $\tr\left[ B \rho^a\rho^b\rho^c\rho^d \right]$, which equals
    \begin{gather*}
        \alpha^- (x^a - x^d)(x^b - x^c) + \beta^- (y^a - y^d)(y^b - y^c) + \gamma^- (z^a -
        z^d)(z^b - z^c)
    \end{gather*}
    which is in turn equal to
    \begin{gather*}
        \left( r^b - r^c \right)^{\top} W' \left(r^a - r^d\right) \quad\text{for}\quad W'= \diag\left(\alpha^- ,\beta^- ,\gamma^- \right).
    \end{gather*}

	If we fix the state of qubits $a$ and $d$ and minimize the expectation on $B$, the minimum is achieved when $r^b = - W'(r^a-r^d) / \norm{W'(r^a-r^d)} $ and $r^c=-r^b$. This minimum expectation is
	$- 2 \norm{W'(r^a-r^d)}$.
	Observe this expectation value is equivalent to
	$-2\norm{W(r^a-r^d)}$ for $W=\diag(\alpha,\beta,\gamma)$ where $\alpha=\abs{\alpha^-},\beta=\abs{\beta^-},\gamma=\abs{\gamma^-}$.

	Now we are prepared to set up our reduction.
	For $W=\diag(\alpha,\beta,\gamma)$, which is non-negative and nonzero,
	consider an arbitrary instance of \Wlinearmaxcut{}.
	For a given graph $G=(V,E)$, the objective function is again $\Wlinearmaxcut(G)$, and given a parameter $C$, we must decide whether $\Wlinearmaxcut(G)$ is at least $C$ or at most $C-\eps$.

    Associate a ``vertex qubit'' with each vertex and construct a copy of the gadget $B$ for each edge $ij$, such that it acts on $i,j$ and two ancilla qubits, and denote it $B^{ij}$. The vertex qubits may be shared among several gadgets, while the ancilla qubits are part of only one gadget.
    In particular, we choose $a$ and $d$ in each gadget to be the vertex qubits.

    Substituting our gadget for every edge constructs a Hamiltonian $H_{\text{final}}$
    such that the minimum expectation of any product state is equal to
    \begin{equation*}
         2 \min_{r^i \in S^2} \sum_{ij\in E} -\norm{Wr^i - Wr^j}
         = - 2 \max_{r^i \in S^2} \sum_{ij\in E} \norm{Wr^i - Wr^j}.
    \end{equation*}
	Simply multiplying $H_{\text{final}}$ by $\frac{1}{2}$ makes this equal to $-\Wlinearmaxcutop(G)$.

    We conclude that given $\calS$ contains a 2-qubit symmetric term $H$ that is not 1-local,
    there exists some non-negative $W=\diag(\alpha,\beta,\gamma)$ with at least one nonzero entry
    such that
    $\Wlinearmaxcut{}(G)$ reduces to $\prodLH{}(H_{\text{final}})$.
    Since $H_{\text{final}}$ was constructed using only the symmetric term $H\in\calS$, this is also a reduction to an instance of \sprodLH{}, as desired.
\end{proof}

\section{The Stretched Linear Max-Cut Problem}\label{sec:vectorProblems}

We study a generalization of the classical \maxcut{} problem which arises naturally from our study of product states and which is likely of independent interest.
Both \maxcut{} and its generalization
\[
	\kmaxcutop(G) =
	\frac{1}{2}\max_{\hat{\imath}\in S^{k-1}}
	\sum_{ij\in E} 1-\hat{\imath}\cdot \hat{\jmath} =
	\frac{1}{4}\max_{\hat{\imath}\in S^{k-1}} \sum_{ij\in E} \norm{\hat{\imath}-\hat{\jmath}}^2
\]
were introduced in \cref{sec:intro}.
As the above equation emphasizes, maximizing the distance between vectors is equivalent to maximizing the angle, optimally being anti-parallel.

Our new problem is
defined with two significant changes.  First, the sum is over distances rather
than squared distances.  Second, the distances are allowed to incorporate a linear stretch.
\begin{definition}[\WlinearmaxcutLong{} (\Wlinearmaxcut)]\label{def:WlinMaxCut}
    For a fixed $d\times d$ diagonal matrix $W$, given an $n$-vertex graph
    $G=(V,E)$ and thresholds $b > a \ge 0$ with $b - a \ge 1/\poly(n)$, decide whether
    \[
        \Wlinearmaxcut(G) = \frac{1}{2} \max_{\hat{\imath} \in S^{d+1}} \sum_{ij\in E} \norm{W\hat{\imath}-W\hat{\jmath}} = \frac{1}{2} \max_{\hat{\imath} \in S^{d+1}} \sum_{ij\in E} \sqrt{\norm{W\hat{\imath}} + \norm{W\hat{\jmath}} - 2(W\hat{\imath})^{\top}(W\hat{\jmath})}
    \]
    is at least $b$ or at most $a$.
\end{definition}

A comparison of the geometric interpretations of \kmaxcut{} and
\Wlinearmaxcut{} was given in \cref{sec:intro}.
A further interpretation comes from treating the edges of the graph as springs or rubber bands.
As explored in \cite{lovasz2019graphs}, the potential energy of a spring is quadratic
in its length, so the \kmaxcut{} value
represents the total potential energy of the system given a particular embedding.
On the other hand, the force or tension of each spring is linear in its length.
So, \Wlinearmaxcut{} gives the total force, tension, or pressure such an arrangement of springs
would apply to the surface of the sphere (or ellipsoid, more generally).

In both problems, the objective is a linear sum, of either the distances
or the inner products. Both problems generalize the
traditional \maxcut{} problem, since when restricted to $\pm 1$
labels, distances are directly proportional to squared distances.
Previous work was likely motivated to focus on squared distances
because approximation algorithms like SDPs naturally apply to inner products but not to square roots of inner products.

Our main theorem concerning \WlinearmaxcutLong{} is the below.

\WlinMaxCutTheorem*

Our hardness proofs do not require any edge weights (unlike our Hamiltonian
constructions in the previous section).

Containment in $\NP$ is immediate, and we break the proof of $\NP$-hardness
into three cases based on the entries of $W$.
The three cases depend on how many entries of $W$ are equal,
requiring different approaches for dealing with degenerate solutions.
We assume throughout that $1 = \alpha \ge \beta \ge \gamma$;
as we show in the final proof of~\Cref{thm:WlinMaxCut},
this suffices by scaling and symmetry.
\Cref{lem:WlMC_allEqual} considers the case when all three entries are equal.
\Cref{lem:WlMC_oneMax} considers the case when the largest entry is unique.
\Cref{lem:WlMC_twoMax} finally considers the case when the two largest entries are
equal and distinct from the third, combining techniques from the previous two proofs.

When $W=\diag(1,1,1)$, we prove hardness by reducing from the $\NP$-complete
\threeColoring{} problem.
We replace every edge in the graph with a 4-clique, or tetrahedron.
To deal with the symmetry created by equally weighted axes,
all of the gadgets are connected to a new \emph{sink} vertex,
removing a degree of freedom.
We then argue that there is an assignment to the new graph that simultaneously
(nearly) maximizes all of these cliques iff the original graph is 3-colorable.

\begin{lemma}\label{lem:WlMC_allEqual}
  For $W=\diag(1,1,1)$, \WlinearmaxcutLong{} is $\NP$-hard.
\end{lemma}

\begin{proof}
  We will reduce from $\threeColoring{}$. Consider an arbitrary graph $G =
  (V,E)$ on $n$ vertices and $m$ edges. We construct $H = (V', E')$ as follows:
  Start with $G$.  For each edge $ij \in E$, add vertices $k_{ij},t_{ij}$ and
  connect $i,j, k_{ij}, t_{ij}$ to form a 4-clique.  Then add a \emph{sink}
  vertex $t$ and an edge $tt_{ij}$ for each $t_{ij}$. $H$ therefore consists of
  $m$ edge-disjoint 4-cliques, each containing one edge from $G$, and $m$
  additional edges from vertices $t_{ij}$ to $t$.

  We claim that if $G$ is 3-colorable, then $\Wlinearmaxcut(H) \geq
  m\Wlinearmaxcut(K_4) + m$.  Conversely, we claim that if $\Wlinearmaxcut(H)
  \ge m\Wlinearmaxcut(K_4) + m-\eps$, for an $\eps = \bOm{1/m^2}$ we will
  choose later, then $G$ is 3-colorable.

  First, suppose $G$ is 3-colorable. We will show how to derive a vector
  assignment to $H$ attaining $m\Wlinearmaxcut(K_4) + m$ from a 3-coloring of
  $G$. Note that $\Wlinearmaxcut(K_4)$ is the maximum perimeter of a
  tetrahedron formed by choosing four vectors on the sphere.

  Let $u,v,w,r$ be vectors corresponding to a regular tetrahedron inscribed in
  the unit sphere, known to achieve the maximum perimeter of any inscribed
  tetrahedron at $4\sqrt{6}$ \cite{maehara2001total}.  We 3-color $G$ (and
  therefore the vertices of $H$ other than $(k_{ij})_{ij \in E}$, $(t_{ij})_{ij
  \in E}$, and $t$) with $u,v,w$, assigning each vertex the vector matching its
  color.  Then for each $ij \in E$, we assign $k_{ij}$ the vector in
  $\{u,v,w\}$ not assigned to $i$ or $j$, and $r$ to $t_{ij}$.  Finally, we
  assign $-r$ to $t$.  By construction, this assignment will yield a value of
  $\Wlinearmaxcut(K_4)$ on each 4-clique gadget, and the edges $tt_{ij}$ will
  each contribute exactly 1.  Thus $\Wlinearmaxcut(H) \geq
  m\Wlinearmaxcut(K_4)+m$, as desired.

  Now we will show the converse.  Suppose there exists an assignment
  $(\hat{k})_{k \in V'}$ of vectors achieving greater than
  $m\Wlinearmaxcut(K_4)+m - \eps$ on $H$.  We construct a $4$-coloring of each
  connected component of $H \setminus \cbracd{t}$ as follows: Choose an edge $ij$
  in the corresponding component of $G$ (note that there is a 1-to-1
  correspondence between components of $G$ and $H \setminus \cbracd{t}$).  Let
  $\hat{\imath},\hat{\jmath},\hat{k}_{ij},\hat{t}_{ij}$ be the vectors assigned
  to the vertices of the corresponding clique in $H$. Our coloring will use these
  four vectors as colors, which we denote as the set $\calC$. We assign each
  vertex $v$ the color corresponding to the element of $\calC$ that is closest to
  the vector assigned to $v$ in the $\Wlinearmaxcut$ assignment. We will show
  that this is a proper coloring, and that it assigns the same color to every
  $t_{ij}$, and therefore gives a 3-coloring of $G$.

  To show that this is a proper coloring, we need to show that every pair of
  adjacent vertices in $H \setminus \cbracd{t}$ are assigned different colors,
  i.e.\ that the closest elements of $\calC$ to the vectors assigned to them in
  the $\Wlinearmaxcut$ assignment are distinct.  As every edge in $H \setminus
  \cbrac{t}$ is contained in some 4-clique corresponding to some edge $xy$ of
  $G$, it will suffice to show the following: for every edge $xy$ in the
  component of $G$ containing $ij$, if $d$ is the smallest number of edges in a
  path in $G$ starting with $ij$ and ending with $xy$, each of $\hat{x}$,
  $\hat{y}$, $\hat{k}_{xy}$ and $\hat{t}_{xy}$ is within $\bO{d\sqrt{\eps}}$ of a
  \emph{different} element of $\calC$.

  First we note that, as $H$ consists of $m$ edge-disjoint 4-cliques and $m$
  other edges, and the maximum any assignment can earn on an edge is $1$, the
  lower bound on the total amount the assignment earns implies that every
  4-clique earns at least $\Wlinearmaxcut(K_4) - \eps$ and the other edges
  $(t_{xy}t)_{xy \in E}$ earn at least $1 - \eps$ each. So by trigonometry we
  immediately have that every $\hat{t}_{xy}$ is within $\bO{\sqrt{\eps}}$ of
  $-\hat{t}$, and therefore within $\bO{\sqrt{\eps}}$ of each other, and
  $\hat{t}_{ij}$ in particular. 

  For the other vertices, we proceed by induction on $d$. We have the desired
  result trivially for $d = 1$, as in this case $ij = xy$. Now suppose it holds
  for $d$. Let $xy$ be the end of a $(d+1)$-edge path starting with $ij$.
  Without loss of generality let $y$ be the vertex of $xy$ earlier in the path,
  and let $z$ be the immediately prior vertex in the path, so by the inductive
  hypothesis $\hat{y}$, $\hat{z}$, $\hat{k}_{yz}$, and $\hat{t}_{yz}$ are
  within $\bO{d\sqrt{\eps}}$ of different elements of $\calC$. Furthermore, as
  both $\cbracd{\hat{x}, \hat{y}, \hat{k}_{xy}, \hat{t}_{xy}}$ and
  $\cbracd{\hat{y}, \hat{z}, \hat{k}_{yz}, \hat{t}_{yz}}$ are vertices of
  tetrahedra with perimeters at least $4\sqrt{6} - \eps$, by
  \cref{lem:tetrahedronNearRegular} the edge lengths of these tetrahedra are
  all in the interval $\brac{\frac{4\sqrt{6}}{6} - \bO{\sqrt{\eps}},
  \frac{4\sqrt{6}}{6} + \bO{\sqrt{\eps}}}$. So as we have already shown that
  $\hat{t}_{yz}$, $\hat{t}_{xy}$, and $\hat{t}_{ij}$ are all within
  $\bO{\sqrt{\eps}}$ of each other, then the criteria of
  \cref{lem:tetrahedronSharedVertices} are satisfied with $ABCD =
  \hat{y}\hat{t}_{xy}\hat{x}\hat{k}_{xy}$ and $AEFG =
  \hat{y}\hat{t}_{yz}\hat{z}\hat{k}_{yz}$ and so $\hat{x}$ and $\hat{k}_{xy}$
  are each within $\bO{\sqrt{\eps}}$ of (different) elements of $\cbracd{\hat{z},
  \hat{k}_{yz}}$. So then the result follows by the triangle inequality.

  We now have that every vertex in $H$ is assigned a vector within
  $\bO{m\sqrt{\eps}}$ of an element of $\calC$, and so taking $\eps$ to be a
  sufficiently small constant times $1/m^2$, we can take these distances to be at
  most $0.1$. Moreover, for any two adjacent vertices, the choice of color will
  be different, and so as by \cref{lem:tetrahedronNearRegular} the distances
  between these four vectors are at least $4/\sqrt{6} - \bO{\eps}$, this implies
  any two adjacent vertices are assigned different colors, and so we have a
  proper 4-coloring of $H \setminus \cbracd{t}$. Finally, note that, as every
  $\hat{t}_{xy}$ is within $\bO{\eps}$ of every other one, this implies all the
  $t_{xy}$ were assigned the same color, and so no vertex in $G$ uses this color.
  So this gives us a proper 3-coloring of $G$.
\end{proof}

Next, in \cref{lem:WlMC_oneMax} we consider the case $W=\diag(1,\beta,\gamma)$
and $1>\beta,\gamma \geq 0$, in which the maximum weight is unique.  The
approach in the proof of \cref{lem:WlMC_allEqual}, reducing from
\threeColoring{}, can in fact be modified to cover this case, but analyzing
triangles inscribed in ellipses instead of circles is more technical.  Instead,
we take a different approach and in the case of a unique maximum (whether
$\beta$ equals $\gamma$ or not) give a reduction from \maxcut{} instead of
\threeColoring{}.

We insert ancilla vertices so that every vertex in the original graph
is the center of a large star.
These star gadgets amplify any deviation from the highest-weight axis such that
any near-optimal solution must approximate a standard 1-dimensional labeling,
as in \maxcut{}.

\begin{lemma}\label{lem:WlMC_oneMax}
	For any $W=\diag(1,\beta,\gamma)$ with $1>\beta \geq \gamma \geq 0$,
	\WlinearmaxcutLong{} is $\NP$-hard.
\end{lemma}

\begin{proof}
  We reduce from the standard $\NP$-complete \maxcut{} problem.  For any graph
  $G=(V,E)$ with $n$ vertices and $m$ edges, we construct a graph $H=(V',E')$
  with $V \subseteq V'$ and $E \subseteq E'$ by, for each $v \in V$, adding
  $K=m^3 n$ ancilla vertices $v_i$, and then adding an edge from each of these
  vertices to $v$ so that $v$ is the center of a $K$-star.  Now $\abs{V'} = n(1
  + K)$ and $\abs{E'} = m + Kn$.

  We claim that, for any $C > 1$ and for large enough $n$, $\maxcutop(G)\geq C$ implies
  $\Wlinearmaxcut(H) \geq C+Kn$ and $\Wlinearmaxcut(H) > C+Kn-1/2$ implies
  $\maxcutop(G)\geq C$.

  First, suppose there is a cut of $G$ with value at least $C$. We construct a
  corresponding assignment of vectors to vertices in $V'$.  First assign the
  vector $(1,0,0)$ to all vertices in $V$ which are labeled $+1$ in $G$ and
  $(-1,0,0)$ to those with labels $-1$.  Then, for every vertex in $v\in V$,
  which by construction is at the center of a star of ancilla qubits in $H$,
  assign the vector opposite the one assigned to $v$ to each of the ancilla
  vertices.  This assignment of vectors gives an objective value of at least
  $C$ on the edges from the original graph and $Kn$ on the edges of the star
  gadgets, and so the $\Wlinearmaxcut$ value of this assignment is $C+Kn$.

  Now suppose there exists an assignment of vectors achieving $\Wlinearmaxcut$
  value greater than $C+Kn-1$ on $H$. We will show that the cut given by
  $\sgn(\hat{v}_x)$ for each $v \in V$, where $\hat{v} = (\hat{v}_x, \hat{v}_y,
  \hat{v}_z)$, (i.e.\ projecting $\hat{v}$ to the $x$-axis and checking whether
  it is $\ge 0$ or $< 0$) has value at least $C$. 

  First, for each $v \in V$, let $\widehat{\sgn}(\hat{v}) =
  (\sgn(\hat{v}_x),0,0)$.  We will show that this is close to $\widehat{v}$.
  Because the original graph can contribute at most $m$ to the $\Wlinearmaxcut$
  objective, and each star gadget can contribute at most $K$, each star gadget
  must contribute at least $K + (C - 1/2 - m) \geq K-1/2-m > K \parend*{1 -
  \frac{2m}{K}}$.  By \cref{lem:stars}, for the $K$-star to achieve at least $K
  \parend*{1 - \frac{2m}{K}}$, the vector $\hat{v}$ assigned to $v$ must
  satisfy
	\[
	\norm{W \hat{v} - \widehat{\sgn}(v)} \leq \delta
	\quad\text{for}\quad
	\delta = 2 \sqrt{\frac{2m}{K}}  \sqrt{\frac{1+\beta^2}{1-\beta^2}} .
	\]
  We will use this fact to show that the $\Wlinearmaxcut$ value earned by the
  vector assignment on the original graph $G$  is close to the value of the cut
  we defined. We have
	\begin{gather*}
		\abs{ \sum_{ij \in E} \norm{W\hat{\imath} - W\hat{\jmath}}
			- \sum_{ij \in E} \norm{\widehat{\sgn}(\hat{\imath})
				-\widehat{\sgn}(\hat{\jmath}) } }
		= \abs{ \sum_{ij \in E} \parend*{\norm{W\hat{\imath} - W\hat{\jmath}}
			- \norm{\widehat{\sgn}(\hat{\imath})-\widehat{\sgn}(\hat{\jmath})}} } \\
		\leq \abs{ \sum_{ij \in E} \parend*{\norm{W\hat{\imath} - \widehat{\sgn}(\hat{\imath})}
			+ \norm{W\hat{\jmath} - \widehat{\sgn}(\hat{\jmath})} }}
		\leq 2m \delta
		= 4 \sqrt{\frac{2}{n}} \sqrt{\frac{1+\beta^2}{1-\beta^2}} = \bO{1/\sqrt{n}}
	\end{gather*}
	which is $<1/2$ for large enough $n$.

  Using for a second time the fact that the edges of the star gadgets can
  contribute at most $Kn$ to the $\Wlinearmaxcut(H)$, the vector assignment
  must achieve at least $C-1/2$ on $G$, and so \[
  C - 1/2 \leq \frac{1}{2} \sum_{ij \in E} \norm{W\hat{\imath} - W\hat{\jmath}}
  \leq \frac{1}{2}\sum_{ij \in E}
  \abs{\widehat{\sgn}(\hat{\imath})-\widehat{\sgn}(\hat{\jmath})} +
  \bO{1/\sqrt{n}}
  \]
  implying that the value of our cut is strictly greater than $C - 1$ for
  sufficiently large $n$.  Therefore, as it is integer-valued it is at least
  $C$, concluding the proof.
\end{proof}

Finally, we give \cref{lem:WlMC_twoMax}, in which $W=\diag(1,1,\gamma)$ and
$1>\gamma$.
Our proof combines the techniques in the previous two proofs.
Similar to the proof of \cref{lem:WlMC_allEqual}, we show hardness by reducing from \threeColoring{}.
Our construction begins by replacing every edge in the graph with a 3-clique.
Then, as in the proof of \cref{lem:WlMC_oneMax},
we insert large star gadgets on every vertex, forcing solutions
away from the low-weight $z$-axis.
With solutions restricted to two dimensions, we are able to argue
there is an assignment to the new graph that simultaneously (nearly) maximizes all of
the cliques iff the original graph was 3-colorable.

\begin{lemma}\label{lem:WlMC_twoMax}
	For any $W=\diag(1,1,\gamma)$ with $1 > \gamma \geq 0$, \WlinearmaxcutLong{} is $\NP$-hard.
\end{lemma}

\begin{proof}
  We will reduce from \threeColoring{}. Consider any graph $G=(V,E)$ with $n$
  vertices and $m$ edges.  Construct a new graph $H'=(V',E')$ by, for each edge
  $ij\in E$, adding a vertex $k_{ij}$ and edges $ik_{ij}$ and $jk_{ij}$, so
  that each edge in $G$ corresponds to a 3-clique ($K_3)$ in $H'$.  Note that
  the $m$ cliques constructed this way are edge-disjoint. Now $H'$ has $n + m$
  vertices and $3m$ edges.  Next, construct $H'' = (V'',E'')$ by, for each
  vertex $v$ in $H'$, add $K=m^{6}$ ancilla vertices $v_i$, each connected to
  $v$ so that $v$ is the center of a $K$-star.  Now $H''$ has $(K+1)(n+m)$
  vertices and $3m+K(n+m)$ edges.

  We claim that if $G$ is 3-colorable, then $\Wlinearmaxcut(H'') \geq
  K(n+m)+3\sqrt{3}m$.  Conversely, we claim that if $\Wlinearmaxcut(H'') >
  K(n+m)+3\sqrt{3}m-\eps$, for an $\eps = \bOm{1/m^2}$ we will choose later,
  then $G$ is 3-colorable.  As testing 3-colorability is $\NP$-hard, this will
  prove the result.

  First, suppose $G$ is 3-colorable.  Let $C$ be any set of three vectors in
  the $xy$-plane achieving the maximum value of
  $\Wlinearmaxcut(K_3)=3\sqrt{3}$.  Given any 3-coloring of $G$, we assign one
  of these vectors to each color, and thus assign vectors from $C$ to $G$ with
  no two adjacent vertices having the same vector.  We extend this assignment
  to $H'$ by, for each of our constructed 3-cliques $i,j,k_{ij}$, assigning the
  vector in $C$ that was not assigned to either $i$ or $j$ to $k_{ij}$.  Now
  each of these cliques contributes $3\sqrt{3}$ to the objective, and as there
  are $m$ of them and they are edge-disjoint, this contributes $3\sqrt{3}m$ to
  the objective.  Finally, we note that each vertex $v_i$ added in the
  construction of $H''$ is incident to exactly one vertex $v$ in $H'$, and no
  other vertices. So we will assign $-\hat{v}$ to $v_i$, where $\hat{v}$ is the
  vertex assigned to $v_i$, contributing $\norm{W\hat{v}} = 1$ to the objective
  value, as $\hat{v}$ is a length-$1$ vector in the $xy$-plane by construction.
  As there are $(n + m)$ vertices in $H'$, the edges added in constructing
  $H''$ will therefore contribute $K(n + m)$ to the objective.  So in total,
  there exists an assignment with value $3\sqrt{3}m+K(n+m)$.

    Now suppose there exists a vector assignment achieving greater than
    $K(n+m)+3\sqrt{3}m-\eps$ on $H''$.  For a vector
    $\hat{v}=(\hat{v}_x,\hat{v}_y,\hat{v}_z)$, let $\widehat{\sgn}(\hat{v})$ denote
    the vector $(\hat{v}_x,\hat{v}_y,0)/\norm{(\hat{v}_x,\hat{v}_y,0)}$.  We assign
    colors as follows.  For each connected component of $H'$, choose any of the
    $K_3$ gadgets $i,j,k_{ij}$ in the component corresponding to an edge $ij$ in
    the original graph.  Let $\hat{\imath},\hat{\jmath},\hat{k}_{ij}$ be the
    vectors assigned to the vertices, respectively. These will be our colors; we
    will call the set of them $\calC$. Each vertex $v$ in the component will be
    assigned whichever vector in $\calC$ is closest to $\widehat{\sgn}(\hat{v})$ as
    its color. We will show that  this gives a proper coloring of $H'$ and therefore
    of $G$.

    First we will show that for every vertex $v$ in $H'$, $\hat{v}$ is close to
    $\widehat{\sgn}(\hat{v})$. Then we will use this fact to show that the coloring
    above is proper. The objective value of the vector assignment consists of what
    is earned by the $m$ 3-cliques in $H'$, along with the $(n + m)$ $K$-stars in
    $H''$.  Because the cliques can each contribute at most $3\sqrt{3}$ to the
    objective value, the stars must contribute at least $K(n+m)-\eps$.  Similarly,
    because the edges of each star can contribute at most $K$, each star must
    achieve at least $K-\eps$.  For any vertex $s$ in $H'$, consider the star
    centered on it in $H''$.  As shown in \cref{lem:stars}, for the star to achieve
    $K(1-\eps/K)$, the vector $\hat{s}$ assigned to $s$ must satisfy
	\[
	\norm{W \hat{s} - \widehat{\sgn}(\hat{s})} \leq \delta
	\quad\text{for}\quad
	\delta = 2 \sqrt{\frac{\eps}{K}}  \sqrt{\frac{1+\gamma^2}{1-\gamma^2}} .
	\]
    This in turn implies that the total value earned by the vector assignment
    on $H'$ is close to what would be earned by the rounded vectors
    $\widehat{\sgn}(\hat{s})$.
	\begin{gather*}
		\abs{ \sum_{ij \in E'} \norm{W\hat{\imath} - W\hat{\jmath}}
			- \sum_{ij \in E'} \norm{\widehat{\sgn}(\hat{\imath})
				-\widehat{\sgn}(\hat{\jmath}) } }
		= \abs{ \sum_{ij \in E'} \norm{W\hat{\imath} - W\hat{\jmath}}
			- \norm{\widehat{\sgn}(\hat{\imath})-\widehat{\sgn}(\hat{\jmath})} } \\
		\leq \abs{ \sum_{ij \in E'} \norm{W\hat{\imath} - \widehat{\sgn}(\hat{\imath})}
			+ \norm{W\hat{\jmath} - \widehat{\sgn}(\hat{\jmath})} }
		\leq 2 \abs{E'} \delta = 12 \sqrt{\eps} m^{-2}  .
	\end{gather*}
    Now we will use this fact to prove that the coloring of $H'$ is proper.
    First note that because the stars can contribute at most $K(n+m)$ to the
    objective value, the vector assignment must achieve at least $3\sqrt{3}m-\eps$
    on the remaining edges, the ones in $H'$.  Therefore, for $\mu= \eps + 12
    \sqrt{\eps} m^{-2}$, the set of rounded vectors must achieve at least
    $3\sqrt{3}m-\mu$ on the rest of $H'$.  Because each of the $m$ clique gadgets
    can contribute at most $3\sqrt{3}$, the rounded vectors must achieve at least
    $3\sqrt{3}-\mu$ on each individual clique.

    We can now use a similar approach to the proof of \cref{lem:WlMC_allEqual}
    to show that, for each component of $H'$ and for every edge $uv$ in the
    component, $\widehat{\sgn}(\hat{u})$ and $\widehat{\sgn}(\hat{v})$ have
    different closest members of the set of colors $\calC$ we chose for that
    component.  The rounded vectors exist in the $xy$-plane, which means they are
    inscribed in the unit circle.  Maximizing the sum of the edge lengths on a
    $K_3$ gadget is equivalent to maximizing the perimeter of a triangle.  For a
    triangle in the unit circle to have nearly maximal perimeter, it must be nearly
    regular; as shown in \cref{lem:triangleNearRegular}, if the perimeter is at
    least $3\sqrt{3}-\mu$, then each edge length must be in the interval
    $\sqrt{3}\pm 3\sqrt{\mu}$.

    Now, we will prove by induction on $d$, the number of vertices in the
    shortest path from $v$ to the 3-clique $ijk$ we choose when assigning the colors,
    that every vertex $v$ in the component has $\widehat{\sgn}(\hat{v})$ within
    $\bO{d\sqrt{\mu}}$ of an element of $\calC$.  Combined with the previous fact
    about the perimeter, this will allow us to prove that the endpoints of
    every edge are assigned different colors. Clearly the result holds for $d = 1$,
    as $\widehat{\sgn}(\hat{v})$ is within $12\sqrt{\eps}m^{-2} \le \mu$ of
    $\hat{v} \in \calC$. Now suppose the result holds for $d$, and $u$ be a vertex
    with a length-$(d+1)$ shortest path to $ijk$. There is therefore a
    length-$(d+2)$ path from $u$ that ends in an edge included in $ijk$.  Choose
    vertices $v,w$ on that path so that $uv$ is the last edge of the path and $vw$
    is the previous one.  By the construction of $H'$, $u$ and $v$ are in a
    3-clique $uvx$, while $v$ and $w$ are in a different one $vwy$.  By the
    inductive hypothesis, $\widehat{\sgn}(\hat{v})$, $\widehat{\sgn}(\hat{w})$, and
    $\widehat{\sgn}(\hat{y})$ are each within $\bO{d\sqrt{\mu}}$ of an element of
    $\calC$. Therefore, by \cref{lem:triangleSharedVertices} with $ABC = vwy$ and
    $ADE = vux$, and using the previous observation about the edge lengths,
    $\widehat{\sgn}(\hat{u})$ is within $\bO{\sqrt{\mu}}$ of one of
    $\widehat{\sgn}(\hat{v})$, $\widehat{\sgn}(\hat{w})$, and
    $\widehat{\sgn}(\hat{y})$, and therefore within $\bO{(d+1)\sqrt{m}}$ of an
    element of $\calC$.

    We are now ready to prove that we have a proper coloring. Let $uv$ be any
    edge in $H'$, and let $\calC$ be the set of colors we chose for that component.
    By the above, each of $\widehat{\sgn}(\hat{u})$ and $\widehat{\sgn}(\hat{v})$
    is within $\bO{m\sqrt{\mu}} = \bO{\eps^{1/4} + m\sqrt{\eps}}$ of an element of
    $\calC$. But they are also at least $3\sqrt{3} - \mu = 3\sqrt{3} -
    \bO{\sqrt{\eps}}$ from each other by our earlier observation about the
    perimeter. So as long as we choose $\eps = \bT{m^{-2}}$ with a sufficiently
    small constant, these will always be different elements of $\calC$ and so $u$
    and $v$ will always be assigned different colors.
\end{proof}

We now conclude with a proof of the main theorem of this section, that
\WlinearmaxcutLong{} is $\NP$-complete for any fixed diagonal
$3\times 3$ non-negative nonzero $W$.

\begin{proof}[Proof of \cref{thm:WlinMaxCut}]
	The containment of \WlinearmaxcutLong{} for any diagonal $W$ is straightforward.
	With $W$ a constant, given a claimed vector assignment,
	the value $\Wlinearmaxcut(G)$ can be verified in time linear in the number of edges.

	To show hardness, we make two simplifications.
	First, because $\Wlinearmaxcut[cW] = c \Wlinearmaxcut[W]$ for a constant $c$,
	we can easily reduce to an instance in which we assume the largest entry of
	$W$ equals $1$.
	Second, although rearranging the entries of diagonal $W$ requires changing any
	vector assignment, it does not change the objective value.
	So
	$\Wlinearmaxcut[\diag(\alpha,\beta,\gamma)] =
	\Wlinearmaxcut[\diag(\beta,\alpha,\gamma)]$
	and any other rearrangement of $W$,
	and we can assume the entries are ordered
	$1\geq \alpha \geq \beta \geq \gamma \geq 0$.
	With this, the theorem follows by
	\cref{lem:WlMC_allEqual,lem:WlMC_twoMax,lem:WlMC_oneMax}.
\end{proof}

\section{NP-Hardness of Unweighted Quantum Max-Cut}
\label{sec:qmc}
Our lower bounds elsewhere in this paper are for local Hamiltonian problems in
which terms can be given positive \emph{or} negative weight. They apply in the
model where all weights are restricted to being of constant weight but require
some terms with negative weight to work. In this section we show that this
restriction can be removed for one of the best-studied local Hamiltonian
problems: the \qmcLong{} (or \qmc{}) problem.

The \qmcLong{} problem can be defined as \sLH{} with $\calF = \{XX+YY+ZZ\}$.
We will write $h$ for  $XX+YY+ZZ$.\footnote{Other work frequently includes
a multiplicative factor, e.g.\ $1/2$, in the definition of $h$ and/or of
\qmc{}.}

\begin{definition}[\qmcLong{} (\qmc{})]
	Given a Hamiltonian $H= \sum^n_{ij} w_{ij} H_{ij}$ acting on $n$ qubits where each $H_{ij} = I - h_{ij}$, all $w_{ij}$ are real, polynomially-bounded, and specified by at most $\poly(n)$ bits,
	and two real parameters $b,a$ such that $b-a\geq 1/\poly(n)$,
	decide whether $\lambda_{\max}(H)$ is at least $b$ (\YES{}) or at most $a$ (\NO{}).
\end{definition}

Note that we have written  \qmc{} as a \emph{maximum} eigenvalue problem (with
a flip and shift of the local terms) rather than in terms of the minimum
eigenvalue as for \kLH{} in \cref{def:kLH}. This is to follow the norm in
previous work; note that as both the terms \emph{and} the objective
function are flipped, an instance of the problem defined this way will be
equivalent to an instance of the corresponding \kLH problem with the same
weights.

When minimizing (maximizing) the eigenvalue and the weights are restricted to
be non-negative (non-positive), it is referred to as the anti-ferromagnetic
Heisenberg model.  Flipping the restrictions, e.g.\ minimizing with
non-positive weights, is referred to as the ferromagnetic Heisenberg model. The
latter case is trivial when viewed as an optimization problem (as $\ket{00}$
earns $0$ on the local term, and so the problem is optimized by assigning
$\ket{0}$ to every qubit), so we will be interested in hardness results for the
former.

It is straightforward to verify that for two qubits $a,b$ with \emph{pure} states $\rho^a,\rho^b$,
\[
	\tr\left(\rho^a\otimes \rho^b ~ h \right)
	= r^a\cdot r^b
	= 1- \frac{1}{2} \norm{r^a-r^b}^2,
\]
where $r^a,r^b$ are the corresponding Bloch vectors.  This shows that deciding
\qmc{} restricted to product states, which we denote $\qmcProd$, is equivalent
to the standard \prob{Vector Max-Cut} problem in three dimensions:

\begin{definition}[\kmaxcutLong (\kmaxcut)] \label{def:kmaxcut}
  Given an $n$-vertex graph $G=(V,E)$ and thresholds $b > a \ge 0$ such that
  $b-a\geq 1/\poly(n)$, decide whether
	\[
	\kmaxcutop[k] = \frac{1}{2} \max_{\hat{\imath} \in S^{k-1}} \sum_{ij \in E} \paren{1 - \hat{\imath}\cdot \hat{\jmath}}
	= \frac{1}{4}\max_{\hat{\imath} \in S^{k-1}} \sum_{ij \in E} \norm{\hat{\imath} - \hat{\jmath}}^2
	\]
	is at least $b$ or less than $a$.
\end{definition}

Note that this is different from the \WlinearmaxcutLong{} we studied in
\cref{sec:vectorProblems} as it considers \emph{squared} distances.
Furthermore, while \maxcut{} is a classic $\NP$-complete problem, $\kmaxcut[k]$
is not expected to be hard for all values of $k$, and in particular is
tractable when $k=n=\abs{V}$.

Our  main result classifying $\sprodLH$ immediately implies $\sprodLH[\{h\}]$,
and therefore also $\qmcProd$ and $\kmaxcut[3]$, is $\NP$-complete.  However,
the proof of \cref{lem:symmetricProdLH} utilizes Hamiltonian gadgets involving
negative weights.  This leaves open whether \qmcProd{} and \kmaxcut[3] remain
$\NP$-hard on unweighted graphs.  We now prove that \kmaxcut[3] is $\NP$-hard
even when restricted to positive unit weights. This is the first published
proof of this fact, although we note that a sketch of a different proof was
previously known for this specific problem~\cite{johnWright}.

Our approach is similar to the proof of \cref{lem:WlMC_allEqual}, which demonstrated
hardness of \Wlinearmaxcut{} for $W=I$ by replacing every edge with a
4-clique (with one vertex connected to a source vertex), and showing that the resulting graph could simultaneously optimize all of these 4-cliques (by assigning vectors corresponding to a regular tetrahedron) if and only if the original graph was 3-colorable.

However, the change in objective function from distances to \emph{squared}
distances causes a problem: while the regular tetrahedron is the unique optimal
solution for $\Wlinearmaxcut(K_4)$, in the case of $\kmaxcut[3](K_4)$, setting
$v_1=v_2=(1,0,0),v_3=v_4=(-1,0,0)$ would also be optimal. So we replace every
edge of the 4-cliques with triangles, which penalize the degenerate solution
and force the vectors assigned to the vertices of the 4-cliques toward regular
tetrahedra.

\maxCutthreeNPcomplete*
\begin{proof}[Proof of \cref{thm:3MaxCutNPcomplete}]
  Clearly $\kmaxcut[3]$ is in $\NP$.  To show hardness, we reduce
  $\threeColoring$ to $\kmaxcut[3]$.  Given a graph $G = (V,E)$ on $n$ vertices
  and $m$ edges, we first construct $H' = (V,E)$ by replacing every edge in $E$
  with a copy of $K_4$.  For each edge $ij\in E$, we add vertices $q_{ij},t_{ij}$
  and add edges to form a 4-clique.  Now we construct $H''$ from $H'$ by
  replacing every edge in $E'$ with a copy of $K_3$.  For each $ij\in E'$ we add
  a vertex $r_{ij}$ and edges $ir_{ij}$ and $jr_{ij}$.  Finally, we add a
  \emph{sink} vertex $t$ and for every edge $ij\in E$ (that is, in the original
  graph), add edge $tt_{ij}$.  Let $R$ denote a copy of $K_4$ with each edge
  replaced by a copy of $K_3$, which we may call a ``tetrahedron with adjoined
  triangles''. $H''$ now consists of a copy of $R$ for every edge $uv$ in $G$,
  with $u,v$ as two of the $K_4$ vertices, and a sink vertex with an edge to each
  $K_4$, incident to one of the vertices that was \emph{not} in the original
  graph $G$.

  We claim that if $G$ has a proper 3-coloring then $\kmaxcut[3](H'') \geq
  m\kmaxcutop[3](R)+m$.  Conversely, we claim that if $\kmaxcut[3](H'') \geq m
  \kmaxcutop[3](R)+m - \eps$, for an $\eps = \bOm{1/m^2}$ we will choose later,
  then there is a 3-coloring of $G$.  By \cref{lem:starTetrahedronNearRegular},
  $\kmaxcutop[3](R)$ is $\frac{1}{4} (40 + 8\sqrt{3}) = 10+2\sqrt{3}$.

  First, suppose $G$ is 3-colorable. Fix a 3-coloring.
  Let $S$ consist of the following three unit vectors in $\mathbb{R}^3$:
  \[
  \left(\sqrt{8/9}, 0, -1/3\right),
  \left(-\sqrt{2/9}, \sqrt{2/3}, -1/3\right),
  \left(-\sqrt{2/9}, -\sqrt{2/3}, -1/3\right).
  \]
  Along with $(0,0,1)$, these are the four vertices of a regular tetrahedron
  inscribed in the unit sphere.  We associate each vector in $S$ with a color,
  and assign every vertex in $G$ the vector given by its color. Then, for the
  vectors in $H' \setminus G$ (the vectors $q_{ij}$ and $t_{ij}$ for $ij \in E$),
  we assign each $q_{ij}$ the vector in $S$ not assigned to $i$ or $j$, and
  assign $(0,0,1)$ to $t_{ij}$. Finally, for the vertices in $H'' \setminus H'$,
  we assign $(0,0,-1)$ to $t$ and, for each $ij \in E'$, assign $r_{ij}$ the
  unique vector in the sphere that is antiparallel to the vectors assigned to $i$
  and $j$.

  Now we calculate the objective value that this assignment achieves.  For
  vertex $i$, let $\hat{\imath}$ denote the vector we assign to it.  For any edge
  $ij$ in $G$, the vectors assigned to the associated copy of $K_4$ in $H'$
  correspond to (distinct) vertices of a regular tetrahedron, and it can be
  directly calculated that for any of the six edges $ab$, we have $\hat{a}\cdot
  \hat{b}=-1/3$. So we earn $6m(1 + 1/3)/2 = 4m$ on these edges. 

  Then, for the $12m$ edges $ir_{ij}, jr_{ij}$ we added when creating the
  $K_3$ copies in $H''$, each one goes between a vertex $i$ or $j$ that was
  assigned one of the regular tetrahedron vectors, and a vertex $r_{ij}$ that was
  assigned the vector antiparallel to the midpoint of that tetrahedron vector and
  another tetrahedron vector. By the rotational symmetry of a regular
  tetrahedron, we may assume these tetrahedron vectors were $\paren{-\sqrt{2/9},
  \sqrt{2/3}, -1/3}$ and $\paren{-\sqrt{2/9}, -\sqrt{2/3}, -1/3}$, so
  $\hat{r}_{ij} = \paren{\sqrt{2/9},0,1/3}/\sqrt{2/9 + 1/9} =
  (\sqrt{2/3},0,\sqrt{1/3})$, and so $\hat{\imath}\cdot \hat{r}_{ij} =
  \hat{\jmath} \cdot r_{ij} = -2/\paren{3\sqrt{3}} - 1/\paren{3\sqrt{3}} = -\sqrt{3}/3$. So we
  earn $12m(1+\sqrt{3}/3)/2 = 6m + 2m\sqrt{3}$ on these edges.
  
  Finally, for the $m$ edges $tt_{ij}$, we have $\hat{t}\cdot
  \hat{t}_{ij}=-1$, and so we earn $m$ on these edges. So adding these together,
  the $\kmaxcut[k]$ value of this assignment is $4m + 6m + 2m\sqrt{3} + m =
  m\kmaxcutop[3](R)+m$ as desired.

  Conversely, suppose there is a set of unit vectors $\hat{\imath}$ for $i \in V''$
  such that $\kmaxcut[3](H'') \geq m\kmaxcutop[3](R)+m - \eps$.  We will use
  this to construct a 3-coloring of $G$. For each component of $G$, choose an
  edge $ab$ in that component, and let $\calC$ denote the three vectors
  assigned to the vertices $a,b,r_{ab}$ in $H'$. We will use $\calC$ as our set
  of colors, assigning each vertex $v$ in $G$ the color closest (in Euclidean
  distance) to $\hat{v}$.

  We need to show that for every edge $ij$ in $G$, $i$ and $j$ were assigned
  different colors. First, we will show that the vectors assigned to these
  vertices are close to being the vertices of a regular tetrahedron.  The graph
  $H''$ is comprised of $m$ edge-disjoint copies of $R$ as well as $m$ edges
  $tt_{ab}$. The latter earn at most $m$ in total (as no edge can earn more than
  $1$). Therefore, for the total objective value to be greater than
  $m\kmaxcutop[3](K_3)+m-\eps$, each copy of $R$ must earn at least
  $\kmaxcutop[3](R)-\eps$.  By \cref{lem:starTetrahedronNearRegular}, for a copy
  of $R$ to earn at least $10+2\sqrt{3}-\eps$ (which means having the sum of its
  squared edge lengths be at least $40 + 8\sqrt{3} - 4\eps$), each side of the
  tetrahedron formed by the vectors assigned to its $K_4$ vertices must have a
  length in the interval $4/\sqrt{6}\pm \bO{\sqrt{\eps}}$.

  Next, we show that all the vectors $\hat{t}_{ij}$ are close to each other. As
  each copy of $R$ earns at most $\kmaxcutop[3](K_3)$, each of these edges must
  earn at least $1-\eps$. This means $\norm{\hat{t} - \hat{t}_{ij}}^2 \ge 4 -
  4\eps$, and so $\norm{\hat{t} - \hat{t}_{ij}} = \bO{\sqrt{\eps}}$, and
  therefore all of these vertices are within $\bO{\sqrt{\eps}}$ of one another.

  We will now show that, for every component of $G$ and every edge $ij$ in the
  component, the vectors $\hat{\jmath},\hat{\imath},\hat{r}_{ij}$ are each close
  to some vector in $\calC$. Specifically, if the shortest path from $ij$ to $ab$
  contains $d$ edges, we will show that each of these vertices is assigned a
  vector within $\bO{d\sqrt{\eps}}$ of some vector in $\calC$.

  We proceed by induction on $d$. Clearly this holds for $d = 1$, as then $ij =
  ab$. Now suppose it holds for $d$, and let $ij$ be the end of a length-$(d+1)$
  path starting at $ab$. Let $k$ be the vertex before $ij$ in this path. By
  induction, each of $\hat{k}, \hat{i}$ and $\hat{r}_{ki}$ is within
  $\bO{d\sqrt{\eps}}$ of some vector in $\calC$. Furthermore, $\hat{t}_{ki}$ is
  within $\bO{\sqrt{\eps}}$ of $\hat{t}_{ij}$ by the result in the previous
  paragraph. So using the edge length bound we derived from
  \cref{lem:starTetrahedronNearRegular}, we may apply
  \cref{lem:tetrahedronSharedVertices} with $ABCD =
  \hat{\imath}\hat{t}_{ij}\hat{\jmath}\hat{r}_{ij}$ and $AEFG =
  \hat{i}\hat{t}_{ki}\hat{k}\hat{r}_{ki}$ to conclude that $\hat{j}$ and
  $\hat{r}_{ij}$ are within $\bO{\sqrt{\eps}}$ of (different) members of
  $\cbracd{\hat{k},\hat{r}_{ki}}$, and therefore within $\bO{(d+1)\sqrt{\eps}}$ of
  vectors in $\calC$.

  Now, for any edge $ij \in E$, \cref{lem:starTetrahedronNearRegular} tells us
  that $\hat{\imath}$ and $\hat{\jmath}$ are $4/\sqrt{6} - \bOm{\sqrt{\eps}}$
  from one another, and so as each is within $\bO{m\sqrt{\eps}}$ of some element
  of $\calC$, they have \emph{different} closest elements of $\calC$ as long as
  $\eps$ is chosen to be a sufficiently small constant times $1/m^2$. So $i$ and
  $j$ are assigned different colors, concluding the proof.
\end{proof}

By the correspondence described earlier in the section, we immediately have the
desired bound on $\sprodLH$.
\qmcProdHard*

\section*{Acknowledgements}

Sandia National Laboratories is a multimission laboratory managed and operated by National Technology and Engineering Solutions of Sandia, LLC., a wholly
owned subsidiary of Honeywell International, Inc., for the U.S. Department of Energy’s National Nuclear Security Administration under contract DE-NA-0003525.
This work was supported by the U.S. Department of Energy, Office of Science, Office of Advanced Scientific Computing Research, Accelerated Research in Quantum Computing, Fundamental Algorithmic Research for Quantum Computing. O.P. was also supported by U.S. Department of Energy, Office of Science, National Quantum Information Science Research Centers, Quantum Systems Accelerator. JY was partially supported by Scott Aaronson's Simons Investigator Award.

\newpage
\bibliographystyle{alphaurl}
\bibliography{bibliography.bib}

\newpage
\appendix
\section{Geometry Lemmas}\label{sec:geometryLemmas}

First, we give a simple lemma regarding star graphs.
Our goal is to show that inserting star gadgets into a graph forces
maximal solutions to $\Wlinearmaxcut$ to be close to the highest weighted axes.

\begin{lemma}\label{lem:stars}
	Consider a star graph $S_K$ with center vertex $v$ and $K$ neighbors.
	Consider any $0\leq \eps\leq 1$ and any
	$W=\diag(1,w_2,w_3)$ with $1 \geq w_2 \geq w_3 \geq 0$ and $w_3 < 1$.

	Let $\Pi$ denote the projector onto the axes for which $w_i=1$
	(so the projector onto the $x$-axis or the $xy$-plane).
	Similarly, let $\lambda$ be the largest $w_i$ which is not equal to 1
	(either $w_2$ or $w_3$).
	Let $\widehat{\sgn}(\hat{v}) = \Pi v / \norm{\Pi v}$.

	If vectors are assigned to $S_K$ which achieve at least
	$\Wlinearmaxcutop(S_K) \geq K(1-\eps)$,
	and $\hat{v}$ is the vector assigned to $v$,
	then
	\[
	\norm{W \hat{v} - \widehat{\sgn}(\hat{v})} \leq \delta
	\quad\text{for}\quad
	\delta = 2\sqrt{\eps} \sqrt{\frac{1 + \lambda^2}{1-\lambda^2}} .
	\]
\end{lemma}
\begin{proof}
  Write $v = u + w$, where $u = \Pi v$. Then $\widehat{\sgn}(\hat{v})$ is
  orthogonal to $w$, and so
	\begin{align*}
		\norm{W\hat{v}-\widehat{\sgn}(\hat{v})}^2 &= \norm{W\hat{u} - \widehat{\sgn}(\hat{v})}^2  + \norm{W\hat{w}}^2 \\
    &= \norm{\hat{u} - \frac{\hat{u}}{\norm{\hat{u}}}}^2 + \norm{W\hat{w}}^2\\
    &= (1 - \norm{\hat{u}})^2 + \norm{W\hat{w}}^2\\
    &\le 1 - \norm{\hat{u}}^2 + \lambda^2\norm{\hat{w}}^2\\
    &= (1 + \lambda^2)\norm{\hat{w}}^2 .
	\end{align*}
  Then, the objective value earned by each edge $vx$ is 
	\begin{align*}
		\frac{1}{2}\norm{W(\hat{v}-\hat{x})} & \le \frac{1}{2}\norm{W\hat{v}} + \frac{1}{2}\\
    &= \frac{1}{2}\sqrt{\norm{W\hat{u}}^2 + \norm{W\hat{w}}^2} + \frac{1}{2}\\
    &\le \frac{1}{2}\sqrt{\norm{\hat{u}}^2 + \lambda^2\norm{\hat{w}}^2} + \frac{1}{2}\\
    &= \frac{1}{2}\sqrt{1 - (1 - \lambda^2)\norm{\hat{w}}^2} + \frac{1}{2}\\
    &\le \frac{1}{2}\sqrt{1 - \frac{1 - \lambda^2}{1 + \lambda^2}\norm{W\hat{v}-\widehat{\sgn}(\hat{v})}^2} + \frac{1}{2}\\
    &\le 1 - \frac{1 - \lambda^2}{4(1 + \lambda^2)} \norm{W\hat{v}-\widehat{\sgn}(\hat{v})}^2
	\end{align*}
  with the final line coming from the Taylor expansion of the square root
  function. By summing the objective value across all $K$ edges and comparing
  to our lower bound on $\Wlinearmaxcutop(S_K)$, we have $\eps \ge \frac{1 -
  \lambda^2}{4(1 + \lambda^2)} \norm{W\hat{v}-\widehat{\sgn}(\hat{v})}^2$, and so the result
  follows.
\end{proof}

Next, we study the geometry of triangles.
For a triangle $ABC$, let $\operatorname{L}(ABC)$ denote the sum of the edge lengths.
It is straightforward to show that for a triangle inscribed in a circle
of radius $r$, $\operatorname{L}(ABC)\leq 3\sqrt{3}r$,
which is uniquely achieved by an equilateral triangle
\cite{maehara2001total}.

\begin{lemma}\label{lem:triangleNearRegular}
	Consider a triangle $ABC$ inscribed in the unit circle and any
	$0\leq \eps < 1$.
	If $\operatorname{L}(ABC) \geq 3\sqrt{3} - \eps$,
	then each edge length is in the interval
	$\brac{\sqrt{3} - 3\sqrt{\eps},\sqrt{3} + 3\sqrt{\eps}}$.
\end{lemma}
\begin{proof}
  For the sake of contradiction, suppose $\operatorname{L}(ABC) \geq
  3\sqrt{3}-\eps$ and that there exists an edge length outside the interval
  $\brac{\sqrt{3} - 3\sqrt{\eps},\sqrt{3} + 3\sqrt{\eps}}$. Now, as
  $\operatorname{L}(ABC) \in \brac{3\sqrt{3} - \eps, 3\sqrt{3}} \subseteq
  \brac{3\sqrt{3} - \sqrt{\eps}, 3\sqrt{3}}$, if this edge is shorter than
  $\sqrt{3} - 3\sqrt{\eps}$ there is another edge of length greater than
  $\sqrt{3} + \sqrt{\eps}$, while if it is longer than $\sqrt{3} + 3
  \sqrt{\eps}$ there is an edge of length shorter than $\sqrt{3} -
  \sqrt{\eps}$. So in either case there exists a pair of edge lengths whose
  difference is greater than $4\sqrt{\eps}$.

	We relabel the triangle so that $\norm{AC}\geq \norm{AB} \geq \norm{BC}$ and
	$\norm{AC}-\norm{BC} > 4\sqrt{\eps}$.
	We can rotate the unit circle as necessary so that $ABC$ is of the form
	\[
	A=(a,b) \quad B=(-a,b) \quad C=(c_x,c_y) 
	\]
  with $b \ge 0$. Define $C' = (0, -1) $.  We will show
  $\operatorname{L}(ABC')-\operatorname{L}(ABC) > \eps$, implying
  $\operatorname{L}(ABC)$ is less than $3\sqrt{3}-\eps$, a contradiction.

	Both $\operatorname{L}(ABC)$ and $\operatorname{L}(ABC')$ are sums of 3 edge lengths.
	The edge $AB$ is shared and we have that $\norm{AC'}=\norm{BC'}$,
	so their difference is $2\norm{AC'}-\norm{BC}-\norm{AC}$.
	As $a^2 + b^2 = c_x^2 + c_y^2 = 1$, we have
	\begin{align*}
		\norm{AC'} &= \sqrt{a^2 + (b + 1)^2}\\
               &= \sqrt{2+2b}\\
    \norm{BC}  &= \sqrt{(a + c_x)^2 + (b - c_y)^2}\\
               &= \sqrt{2 + 2ac_x - 2bc_y}\\
		\norm{AC}  &= \sqrt{(a - c_x)^2 + (b - c_y)^2}\\
               &= \sqrt{2 - 2ac_x - 2bc_y}
	\end{align*}
	and so
	\begin{align*}
    \paren{\norm{BC}+\norm{AC}}^2 + \paren{\norm{BC}-\norm{AC}}^2 &=
    2\norm{BC}^2 + 2\norm{AC}^2\\
    &= 8 - 8bc_y\\
    &\le 4\norm{AC'}^2
	\end{align*}
  as $b > 0$ and $\abs{c_y} \le 1$. Putting these together, and using the fact
  that no edge can be longer than $2$, we have
  \begin{align*}
  \operatorname{L}(ABC')-\operatorname{L}(ABC) &= 2\norm{AC'}-\norm{BC}-\norm{AC}\\
  &= \frac{4\norm{AC'}^2 - \paren{\norm{BC} + \norm{AC})}^2}{2\norm{AC'} + \norm{BC} + \norm{AC}}\\
  &\ge \frac{\paren{\norm{BC} - \norm{AC}}^2}{8}\\
  &> 2 \eps
  \end{align*}
	completing the proof.
\end{proof}

\begin{lemma}\label{lem:triangleSharedVertices}
	Consider triangles $ABC$ and $ADE$ inscribed in the unit circle.
	If all edges of the triangle have lengths in the interval
	$\sqrt{3}\pm \delta$,
	then the points $\{B,C\}$ are each within $\bO{\delta}$
	distance of (different) points in $\{D,E\}$.
\end{lemma}
\begin{proof}
  Because these are vectors restricted to a unit circle, we can assume $A =
  (1,0,0)$.  Given $A$ is fixed, we can characterize the constraints on the
  other points as follows.  $B,C,D,E$ must lie in the intersection of the unit
  circle around the origin with an annulus centered at $A$, with radii
  $\sqrt{3}- \delta$ and $\sqrt{3}+ \delta$.  We may assume $\sqrt{3} + \delta
  < 2$ (as otherwise the lemma holds trivially), and so this intersection will
  consist of two disjoint segments of the unit circle, each starting at a point
  $\sqrt{3} - \delta$ from $A$ and ending at a point $\sqrt{3} + \delta$ from $A$.

  For both regions, we want to bound the maximum distance between two points in
  the region.  This distance is the length of a chord from the point closest to
  $A$ to the point farthest from $A$.

  For a chord of length $d$, the internal angle is $2\arcsin(d/2)$; conversely,
  if its internal angle is $\theta$, its length is $2\sin(\theta/2)$.
  Therefore, the internal angles of the chords connecting $A$ to the closest
  and furthest point of a given segment are $2\arcsin((\sqrt{3} \pm
  \delta)/2)$. By the Taylor series of the $\arcsin$ function, the internal
  angle of the chord between the furthest and most distant point of the segment
  is therefore \[
    2\arcsin(\sqrt{3} + \delta) - 2\arcsin(\sqrt{3} - \delta) = \bO{\delta}
  \]
  and so the distance between those points is $\bO{\delta}$.

  Again using the fact that the lemma holds trivially when $\delta$ is lower
  bounded by a constant, we may assume that this distance is smaller than
  $\sqrt{3} - \delta$, and so $B,C$ are in different segments from each other,
  as are $D, E$, and so as any two points in the same segment are within
  $\bO{\delta}$ of each other the lemma follows.
\end{proof}

Next, we transition to considering tetrahedra.
For a tetrahedron $ABCD$,
let $\operatorname{L}(ABCD)$ denote the sum of the edge lengths.
It is known that for a tetrahedron inscribed in a sphere of radius $r$,
$\operatorname{L}(ABCD)\leq 4\sqrt{6} r$,
and the maximum is uniquely achieved by a regular tetrahedron \cite{maehara2001total}.

\begin{lemma}\label{lem:tetrahedronNearRegular}
	Consider a tetrahedron $ABCD$ inscribed in the unit sphere and any $\eps \geq 0$.
	If $\operatorname{L}(ABCD) \geq 4\sqrt{6} - \eps$,
	then each edge length is in the interval $\brac{\frac{4\sqrt{6}}{6} -
4\sqrt{\eps},\frac{4\sqrt{6}}{6} + 4\sqrt{\eps}}$.
\end{lemma}
\begin{proof}
	For the sake of contradiction, suppose $\operatorname{L}(ABCD) \geq 4\sqrt{6}
	- \eps$ and that there exists an edge $e$ with length outside the given
	interval. First we will show that there is some pair of edges in the
	tetrahedron that share a vertex and differ in length by at least
	$3\sqrt{\eps}$. First consider the four edges sharing vertices with $e$. If any
	of them has a length differing by at least $3\sqrt{\eps}$ from the length of
	$e$ we have our pair. Otherwise, as $\operatorname{L}(ABCD) \in \brac{4\sqrt{6}
	- \eps, 4\sqrt{6}} \subseteq \brac{4\sqrt{6} - \sqrt{\eps}, 4\sqrt{6}}$, the
	final edge must have length differing by more than $6\sqrt{\eps}$ from $e$ in
	order for the sum to balance. But this means that one of the edges incident to
	both this edge and $e$ will differ in length from it by at least
	$3\sqrt{\eps}$.

	We can therefore relabel the tetrahedron so that $\norm{AC}-\norm{BC}>
	3\sqrt{\eps}$.  Next, without changing any lengths, we can rotate $ABCD$ such
	that $A,B$ are of the form $A=(a,-b,0)$, $B=(-a,-b,0)$.  Let
	$C=(c_x,c_y,c_z)$, $D=(d_x,d_y,d_z)$ and then define $C' = (0, c_y,
	\sqrt{1-c_y^2})$ and $D'=(0,d_y,-\sqrt{1-d_y^2}) $, We will show
	$\operatorname{L}(ABC'D')-\operatorname{L}(ABCD) > \eps$, implying
	$\operatorname{L}(ABCD)$ must be less than $4\sqrt{6}-\eps$.

	As $\norm{AC'} = \norm{BC'}$ and $\norm{AD'} = \norm{BD'}$, \[
  \operatorname{L}(ABCD) - \operatorname{L}(ABC'D') = \norm{AC} + \norm{AD} +
  \norm{BC} + \norm{BD} + \norm{CD} - 2\norm{AC'} - 2\norm{AD'} - \norm{C'D'}.
  \]
	We start by noting that
  \begin{align*}
  \norm{C'D'} &= \sqrt{(c_y -d_y)^2 + \paren{\sqrt{1 - c_y^2} + \sqrt{1 - d_y^2}}^2}\\
  &= \sqrt{c_x^2 + d_x^2 + (c_y -d_y)^2 + c_z^2 + d_z^2 + 2\sqrt{1 - c_y^2}\sqrt{1 - d_y^2}}\\
  &= \sqrt{(c_x - d_x)^2 + (c_y - d_y)^2 + (c_z - d_z)^2 + 2c_xd_x + 2c_zd_z + 2\sqrt{c_x^2 + c_z^2}\sqrt{d_x^2 + d_z^2}}\\
  &\ge \sqrt{(c_x - d_x)^2 + (c_y - d_y)^2 + (c_z - d_z)^2}\\
  &= \norm{CD}
  \end{align*}
  where the inequality follows from Cauchy-Schwarz. Now, we calculate
  \begin{align*}
    \norm{AC'} &= \sqrt{a^2 + (b + c_y)^2 + 1 - c_y^2}\\
               &= \sqrt{2+2bc_y}\\
    \norm{BC}  &= \sqrt{(a + c_x)^2 + (b + c_y)^2 + c_z^2}\\
               &= \sqrt{2 + 2ac_x + 2bc_y}\\
    \norm{AC}  &= \sqrt{(a - c_x)^2 + (b + c_y)^2 + c_z^2}\\
               &= \sqrt{2 - 2ac_x + 2bc_y}
	\end{align*}
	and likewise
	\begin{align*}
    \norm{AD'} &= \sqrt{2+2bd_y}\\
    \norm{BD}  &= \sqrt{2 + 2ad_x + 2bd_y}\\
    \norm{AD}  &= \sqrt{2 - 2ad_x + 2bd_y} .
  \end{align*}
	We will now proceed on similar lines to the proof of \cref{lem:triangleNearRegular}. We have
  \begin{align*}
    \paren{\norm{BC}+\norm{AC}}^2 + \paren{\norm{BC}-\norm{AC}}^2 &=
    2\norm{BC}^2 + 2\norm{AC}^2\\
    &= 8 + 8bc_y\\
    &= 4\norm{AC'}^2 .
  \end{align*}
	Putting these together, and using the fact that no edge can be longer than
	$2$, we have
	\begin{align*}
		2\norm{AC'}-\norm{BC}-\norm{AC} &= \frac{4\norm{AC'}^2 - \paren{\norm{BC} +
		\norm{AC})}^2}{2\norm{AC'} + \norm{BC} + \norm{AC}}\\
		& \ge \frac{\paren{\norm{BC} - \norm{AC}}^2}{8}\\
		& \ge 9\eps/8
	\end{align*}
  and by the same argument,
	\begin{align*}
		2\norm{AD'}-\norm{BD}-\norm{AD} &\ge \frac{\paren{\norm{BD} - \norm{AD}}^2}{8}\\
		& \ge 0
	\end{align*}
	and so putting the three inequalities together, we conclude that
	$\operatorname{L}(ABC'D')-\operatorname{L}(ABCD) > \eps$, contradiction.
\end{proof}

\begin{lemma}\label{lem:tetrahedronSharedVertices}
	Consider tetrahedra $ABCD$ and $AEFG$ inscribed in the unit sphere
	with $\norm{B-E}\leq \eps$.
	If all edges  of the tetrahedra have lengths in the interval
	$\frac{4\sqrt{6}}{6}\pm \delta$,
	then the points $\{C,D\}$ are each within $\bO{\delta+\eps}$ distance of
	(different) points in $\{F,G\}$.
\end{lemma}

\begin{proof}
  Fix the points $A$, $B$, and $E$. Note that there are two points on the
  sphere, $p_1$ and $p_2$, that are \emph{exactly} $\frac{4\sqrt{6}}{6}$ from
  both $A$ and $B$. By the triangle inequality and the bounds on the edge
  lengths, each of $C$, $D$, $F$, and $G$ are within $\delta + \eps$ of one of
  these points. Let $S_1$, $S_2$ be the $(\delta + \eps)$-balls around $p_1$,
  $p_2$ respectively, so that each of these points is in (at least) one of
  them. 

  The result holds trivially if $2\delta + \eps \ge \frac{4\sqrt{6}}{6}$, so
  assuming this is not the case, the lower bound on the edge lengths guarantees
  that $C$ and $D$ are in different choices of $S_1$, $S_2$, and likewise for
  $F, G$. So then as as $S_1$ and $S_2$ each have a diameter of $2\delta +
  2\eps$ the result follows.
\end{proof}

We will use the name ``tetrahedron with adjoined triangles'' to refer to a
three-dimensional tetrahedron such that for each edge there is an additional
point with which the ends of the edge are joined to form a triangle.  In total,
there are 10 vertices and 12 edges.  The shape will be assumed to be inscribed
in the unit sphere, so a vector describing any vertex of the tetrahedron or a
triangle is a unit vector.  For some instance $R$ of this shape with assigned
vertex locations, let $\operatorname{L}^2(R)$ denote the sum of the
\emph{squared} edge lengths.  The following lemma says that for
$\operatorname{L}^2(R)$ to be maximized, the edges of the tetrahedron (but not
the edges of the adjoined triangles) must approximate a regular tetrahedron.

\begin{lemma}\label{lem:starTetrahedronNearRegular}
  Given a tetrahedron with adjoined triangles $R$, $\operatorname{L}^2(R)
  \leq 40+8\sqrt{3}$, with equality iff the tetrahedron is regular.  Furthermore,
  if any edge of the tetrahedron has a length outside of the interval
  $\brac*{\frac{16}{\sqrt{6}}-\eps,\frac{16}{\sqrt{6}}+\eps}$, then
  \[
  	\operatorname{L}^2(R) < 40+8\sqrt{3} - \bOm{\eps^2} .
  \]
\end{lemma}

\begin{proof}
  Suppose the vertices of the tetrahedron are $V_1,V_2,V_3,V_4$,
  and the additional vertex forming a triangle with edge $V_iV_j$ is $v_{ij}$.
  Let $s_{ij}$ denote the length of edge $ij$.

  Fixing any $V_i$, $V_j$, the sum of the squared lengths of the triangle
  $V_i$, $V_j$, $v_{ij}$ is $s_{ij}^2$ plus the sum of the squared distances from
  $v_{ij}$ to $V_i$, $V_j$. The latter sum is maximized for fixed $V_i$, $V_j$ by
  the point $v_{ij}$ that maximizes the distance from the centroid $(V_i +
  V_j)/2$. This choice of $v_{ij}$ is (treating the sphere as centered at the
  origin) proportional to $-(V_i + V_j)$, i.e.\ it is on the opposite side of the
  sphere from the centroid of $V_i$ and $V_j$. Therefore, by Pythagoras the
  optimal sum of squared lengths is $s_{ij}^2 + 2(h^2 + (s_{ij}/2)^2)$, where $h$
  is the distance from $(V_i + V_j)/2$ to this point $v_{ij}$.  Then, again by
  Pythagoras, the distance between $(V_i + V_j)/2$ and the origin is $\sqrt{1 -
  s_{ij}^2/4}$, and so $h = 1 + \sqrt{1 - s_{ij}^2/4}$. So the optimal value of
  the sum of the squared lengths of the triangle $V_i$, $V_j$, $v_{ij}$ is given
  by \[
    t(s_{ij}) = \frac{3}{2}s_{ij}^2 + 2 + 4\sqrt{1 - s_{ij}^2/4} + 2 -
    \frac{s_{ij}^2}{2} = 4 + s_{ij}^2 + \sqrt{16 - 4s_{ij}^2} .
  \]
  Now, because $R$ is the disjoint union of the six triangles which each have
one of the tetrahedron's edges as a side, if we fix the tetrahedron edge
lengths and then choose the other vertices $v_{ij}$ optimally, we will have
  \[
  	\operatorname{L}^2(R) = \sum_{1 \le i < j \le 4} t(s_{ij}) .
  \]
  A regular tetrahedron inscribed in the unit sphere has edge lengths equal to
  $4/\sqrt{6}$. So, if $R$ is formed by taking a regular tetrahedron and then
  choosing the ideal points for the triangles, then $\operatorname{L}^2(R) = 6
  t (4/\sqrt{6}) = 40+8\sqrt{3}$.

  Now we will consider the case where the tetrahedron is irregular. First we
  note that the first and second derivatives of $t$ are given by
  \begin{align*}
    t'(s_{ij}) &= 2s_{ij}\paren{1 - \frac{1}{\sqrt{4 - s_{ij}^2}}}\\
    t''(s_{ij}) &= 2 - \frac{2}{\sqrt{4 - s_{ij}^2}} - \frac{2s_{ij}^2}{\paren{4 -
    s_{ij}^2}^{3/2}}\\
    &= 2 - \frac{8}{\paren{4 - s_{ij}^2}^{3/2}}
  \end{align*}
  giving us $t'(s_{ij}) > 0$ in the interval $\brac{0, \sqrt{3}}$ (and so $t$
  is increasing in that range) and $t''(s_{ij}) < 0$ in the interval $\brac{1.22,
  2}$ (and so $t$ is concave in that range). Furthermore, as the only zero of
  $t'$ is at $\sqrt{3}$, the maximum of $t$ in the range $\brac{0,2}$ (which is
  all possible lengths in the sphere) is at $\max(t(0), t(\sqrt{3}), t(2)) =
  t(\sqrt{3}) = 9$.

  Now, we can use these facts to bound $\operatorname{L}^2(R)$. First, consider
  the case where some edge length is less than $1.22$. By the above two facts, \[
    \operatorname{L}^2(R) \le t(1.22) + 5t(\sqrt{3}) < 8.7 + 45 <
    53.8 < 40 + 8\sqrt{3}
  \]
  and so the lemma holds in this case. Therefore, we can assume that every edge
  length is in the range where $t$ is concave. 

  Now, suppose that every edge has length at least $1.22$ and there is some
  edge length outside of $\brac{4/\sqrt{6} - \eps, 4/\sqrt{6} + \eps}$.
  By~\cite{maehara2001total}, the maximum sum of the (unsquared) edge-lengths of
  a tetrahedron is $4\sqrt{6}$, and so there is some edge length $s_{ab} <
  4/\sqrt{6} - \eps/5$. We will now show that maximising $\sum_{1 \le i < j \le
  4} t(s_{ij})$ under the constraint $\sum_{1 \le i < j \le 4} s_{ij} \le
  4\sqrt{6}$ gives a value of at most $40 + 8\sqrt{3} - \bOm{\eps^2}$, proving
  the lemma.

  We will start by showing that, without loss of generality, we may assume that
  $\sum_{1 \le i < j \le 4} s_{ij} - s_{ab} = 4\sqrt{6} - s_{ab}$. Because
  $s_{ab} \ge 1.22$, we know there is some other edge with length at most
  $(4\sqrt{6} - 1.22)/5 < 1.72 < \sqrt{3}$, and so it is in the range where $t$
  is increasing. Therefore, increasing that length towards $\sqrt{3}$ would
  increase $\sum_{1 \le i < j \le 4} t(s_{ij})$ (this may not correspond to a
  physically possible tetrahedron, but as we are proving an upper bound, this is
  not a problem).

  So we want to bound $\sum_{1 \le i < j \le 4} t(s_{ij})$ with a fixed $s_{ab}
  < 4/\sqrt{6} - \eps/5$, subject to $\sum_{1 \le i < j \le 4} s_{ij} =
  4\sqrt{6}$. As we already have that every $s_{ij} \ge 1.22$, and so we are in
  the range where $t$ is concave, this is maximized by setting every length other
  than $s_{ab}$ to be equal, giving \[
    \sum_{1 \le i < j \le 4} t(s_{ij}) = t(s_{ab}) + 5 t((4\sqrt{6} - s_{ab})/5).
  \]
  Now, as $s_{ab} \in (1.22, 4/\sqrt{6} - \eps/5)$  we can
  use the fact that  $t''$ is negative in the range $(1.22, 2)$ to conclude that
  $t'(x) > t'((4\sqrt{6} - x)/5)$ for all $x$ in that range, and therefore \[
    t(s_{ab}) + 5 t((4\sqrt{6} - s_{ab})/5) \le t(4/\sqrt{6} - \eps/5) +
    5t(4\sqrt{6} + \eps/25)
  \]
  and so, using the Taylor expansion of $t$ about $4/\sqrt{6}$, we can bound
  \begin{align*}
    \operatorname{L}^2(R) &\le t(4/\sqrt{6} - \eps/5) + 5t(4\sqrt{6} + \eps/25) \\
    &= t(4/\sqrt{6}) - \eps t'(4/\sqrt{6})/5 + \eps^2t''(4/\sqrt{6}) / (2 \times
    5^2) + \bO{\eps^3}\\
    &\phantom{=}+ 5\paren{t(4/\sqrt{6}) + \eps t'(4/\sqrt{6})/25 + \eps^2
    t''(4/\sqrt{6})/(2
    \times 25^2)} + \bO{\eps^3}\\
    &= 6t(4/\sqrt{6}) - \bOm{\eps^2} + \bO{\eps^3}
  \end{align*}
  as $t''(4/\sqrt{6}) < 0$. So there is some constant $C$ such that \[
    \operatorname{L}^2(R) < 40 + 8\sqrt{3} - \bOm{\eps^2}
  \]
  for all $\eps \le C$. Finally, for $\eps$ in the range $\brac{C, 1}$, and
  again using the fact that $t'(x) > t'((4\sqrt{6} - x)/5)$ in the range
  $(1.22,2)$, we have $\operatorname{L}^2(R) \le t(4/\sqrt{6} - C/5) +
  5t(4\sqrt{6} + C/25) < 40 + 8\sqrt{3} - \bOm{1}$, completing the proof of the
  lemma.
\end{proof}

\end{document}